\documentclass[journal]{IEEEtran}
\usepackage{cite}
\usepackage{amsmath,amssymb,amsfonts,amsthm}
\usepackage{algorithm}
\usepackage{algorithmic}
\usepackage{graphicx}
\usepackage{textcomp}
\usepackage{xcolor}
\usepackage{tabularx}
\usepackage{subfig}
\usepackage[margin=0.2cm]{caption}
\usepackage{comment}
\usepackage{booktabs}
\def\BibTeX{{\rm B\kern-.05em{\sc i\kern-.025em b}\kern-.08em
    T\kern-.1667em\lower.7ex\hbox{E}\kern-.125emX}}

\usepackage{bm}
\newtheorem{theorem}{Theorem}

\newtheorem{assumption}{Assumption}

\usepackage{amsmath}

\let\originalleft\left
\let\originalright\right
\renewcommand{\left}{\mathopen{}\mathclose\bgroup\originalleft}
\renewcommand{\right}{\aftergroup\egroup\originalright}

\begin{document}

\title{PreHO: Predictive Handover for LEO Satellite Networks}

\author{
  Xingqiu He,
  Zijie Ying,
  Chaoqun You,
  Yue Gao,~\IEEEmembership{Fellow,~IEEE}
\thanks{X. He, Z. Ying, C. You, and Y. Gao are with the Institute of Space Internet, Fudan University, Shanghai,
China (emails: hexqiu@gmail.com, zjying23@m.fudan.edu.cn, chaoqunyou@gmail.com, gao.yue@fudan.edu.cn).}
}

\maketitle

\begin{abstract}
  Low-Earth Orbit (LEO) Satellite Networks (LSNs) offer a promising solution for extending connectivity to areas not covered by Terrestrial Networks (TNs). 
  However, the rapid movement, broad coverage, and high communication latency of LEO satellites pose significant challenges to conventional handover mechanisms, resulting in unacceptable signaling overhead and handover latency. 
  To address these issues, this paper identifies a fundamental difference between the mobility patterns in LSNs and TNs: users are typically stationary relative to the fast-moving satellites, and channel states in LSNs are often stable and predictable. 
  This observation enables handovers to be planned in advance rather than triggered reactively. 
  Motivated by this insight, we propose PreHO, a predictive handover mechanism tailored for LSNs that proactively determines optimal handover strategies, thereby simplifying the handover process and enhancing overall efficiency. 
  To optimize the pre-planned handover decisions, we further formulate the handover planning problem and 
  develop an efficient iterative algorithm based on alternating optimization and dynamic programming. 
  Extensive evaluations driven by real-world data demonstrate that PreHO significantly outperforms traditional handover schemes in terms of signaling overhead, handover latency, and user experience.
\end{abstract}


\section{Introduction} \label{section:introduction}
The Internet has become an indispensable component of modern information society. 
However, due to the high deployment and operational costs associated with Terrestrial Networks (TNs), 
there remain approximately 2.6 billion ``unconnected'' users worldwide, 
primarily residing in remote and economically underdeveloped regions \cite{itu2023}. 
Given the inherent global coverage of satellites, extending mobile networks into space presents a promising solution. 
While early satellite communication systems predominantly relied on Geostationary Equatorial Orbit (GEO) satellites, 
recent developments have shifted toward Low-Earth Orbit (LEO) constellations, 
attributed to their lower production and launch costs, reduced propagation latency, and diminished path loss \cite{osoro2021techno}.

One critical challenge in LEO Satellite Networks (LSNs) arises from the high velocity of LEO satellites (approximately 7.6 km/s), 
which necessitates frequent handovers among satellites.
Existing handover mechanisms, originally designed for TNs, rely on fluctuations in received signal strength. 
As shown in Fig. 1a, User Equipments (UEs) in TNs experience notable variations in signal strength when approaching or moving away from base stations, 
prompting handover decisions when a neighboring cell's signal becomes sufficiently stronger.

\begin{figure}[t]
\centering
\includegraphics[width=0.45\textwidth]{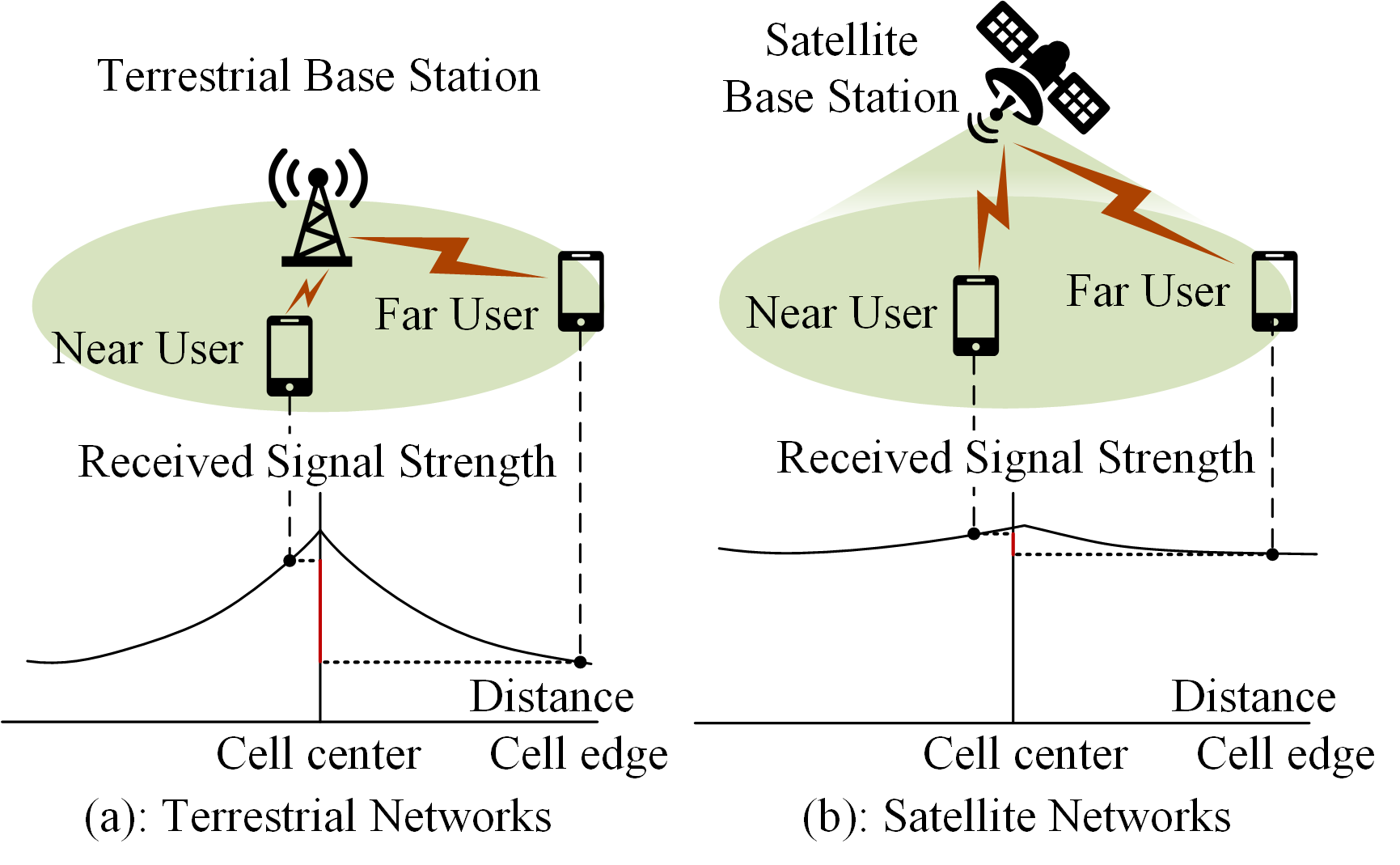}
\caption{Variation of signal strength in TNs and LSNs.}
\label{fig:signal_strength}
\end{figure}

In contrast, LSNs exhibit distinct characteristics that undermine the effectiveness of signal strength-based handovers. 
As illustrated in Fig. 1b, the high altitude of LEO satellites leads to minimal signal strength variation across large areas, 
making it difficult to derive appropriate handover decisions \cite{3gpp.38.821}.
Additionally, this mechanism also imposes other significant limitations in LSNs, 
such as excessive signaling overhead and prolonged handover latency, which will be analyzed in detail in Section \ref{subsection:problems}.

The root cause of the aforementioned issues lies in the fundamentally different mobility patterns of TNs and LSNs. 
In TNs, the random movement of users makes handovers unpredictable, which limits handover efficiency in various aspects.
For example, since it is unknown when and to which base station a user will handover, 
the downlink path switching at the Core Network (CN) cannot be preemptively prepared, hence leading to prolonged handover latency.
Additionally, the target base station can only be determined after the initiation of handover procedure.
Informing both the user and the CN of the relevant information about the target base station requires
multiple signaling messages, which results in significant signaling overhead when the handover frequency is high.
While such inefficiencies are tolerable in TNs, they become unacceptable in LSNs due to the rapid movement, 
broad coverage, and high communication latency of LEO satellites.

Conversely, in LSNs, users are typically considered relatively stationary compared to the rapid movement of satellites \cite{wang2022seamless}. 
Moreover, the users of LSNs are usually located in sparsely populated regions with favorable line-of-sight communication conditions and minimal interference, 
hence their channel states are often stable and predictable \cite{9439942, 10550141, 9990587}. 
Consequently, the randomness associated with users' locations and link quality is considerably mitigated, rendering handovers largely predictable.

Building on this insight, we propose PreHO, a predictive handover mechanism designed for LSNs. 
To facilitate handover planning, we introduce a novel network function, the Handover Planning Function (HPF), 
to collect related data (e.g., user locations) and optimizes handover decisions with global information to enhance user experience. 
These decisions are then disseminated to UEs, satellites, and the CN, 
enabling autonomous execution without further signaling interactions.
Furthermore, PreHO maintains compatibility with conventional handover mechanisms and can seamlessly revert to 
signal strength-based approaches when the predicted information is inaccurate.

Our main contributions are summarized as follows:
\begin{itemize}
  \item 
    We propose PreHO, a predictive handover mechanism tailored for LSNs.
    By leveraging the predictability of handovers, PreHO fundamentally addresses the efficiency bottlenecks inherent 
    in existing handover mechanisms, thereby significantly enhancing the performance of LSNs.
  \item 
    To optimize handover decisions at the HPF, we formulate the handover planning problem as a mixed-integer programming problem, 
    which is proven to be NP-hard. 
    By employing the alternating optimization and dynamic programming, we develop an efficient approximate algorithm that is guaranteed to converge to a local optimum.
  \item 
    We build a prototype of PreHO and evaluate it with real-world satellite data. 
    Experimental results demonstrate that PreHO significantly outperforms existing handover mechanisms
    in terms of handover latency, signaling overhead, and user experience.
\end{itemize}

The rest of the paper is organized as follows. 
In Section \ref{section:related_work}, we review related works.
In Section \ref{section:motivation}, we briefly introduce the preliminary and motivation of our work.
Section \ref{section:system} describes the detailed design of PreHO.
The proposed planning algorithm is demonstrated in Section \ref{section:algorithm}.
The evaluation results and conclusions are presented in Section \ref{section:simulation} and Section \ref{section:conclusion}.

\section{Related Work} \label{section:related_work}
This section reviews related work from two perspectives. 
The first focuses on recent efforts to improve handover mechanisms for LSNs, 
while the second examines the design of predictive handover mechanisms for TNs.

\subsection{Handover Mechanisms for LSNs}
Existing research on handover mechanisms in LSNs primarily falls into two categories. 
The first focuses on optimizing handover timing and target selection. 
For instance, \cite{wang2022seamless} optimizes the target satellite candidates for conditional handover (CHO) under service continuity constraints.
The work in \cite{zhang2021network} formulates handovers as a network flow problem and derive strategies by solving for minimum cost and maximum flow. 
A QoE-aware handover mechanism is proposed in \cite{xu2020qoe} that selects target satellites based on predicted service duration and available communication resources. 
The work in \cite{he2020load} employs multi-agent reinforcement learning to minimize the number of handovers while ensuring load balance across satellites.

The second category aims to refine the handover procedure to better accommodate LSNs. 
As demonstrated in \cite{juan20205g,juan2022performance}, traditional handover mechanisms often lead to service discontinuity and unnecessary handovers in LSNs.
To address this, \cite{juan2022location} proposes a location-based CHO strategy that leverages user location and satellite trajectory to guide handover decisions. 
In \cite{li2020user}, the authors introduce a buffering mechanism in which multiple satellites simultaneously cache a user's downlink data to enhance communication reliability during handovers. 
Additionally, \cite{lee2023handover} utilizes deep reinforcement learning to predict future signal strength, thereby eliminating the need for measurement reports.
The works in \cite{wu2024accelerating, wu2025phandover} reduce handover latency by decoupling access and core network interactions based on predicted handover timings.
However, none of the aforementioned studies effectively address the issue of excessive signaling overhead in LSNs.

\subsection{Predictive Handover Mechanisms for TNs}
Predictive handover mechanisms have been extensively investigated in 5G and beyond networks to improve handover efficiency.
The majority of existing approaches aim to optimize the handover timing and target cell selection by leveraging predictive models 
based on historical radio conditions \cite{masri2021machine, lima2023deep}, user trajectories and resource utilization patterns \cite{sun2024pbphs}, 
as well as user equipment measurements \cite{panitsas2024predictive}. 
Several studies focus specifically on mmWave scenarios. 
For instance, \cite{koda2018reinforcement} proposes a reinforcement learning-based framework that predicts human blockage 
using pedestrian location and velocity to derive optimal handover decisions. 
Similarly, \cite{palacios2019leap} utilizes device positions to optimize network operations 
such as access point association and antenna beam steering.
In \cite{lee2020prediction}, a novel deep learning-based CHO scheme is proposed
by predicting future signal blockages and identifying the optimal target base station based on evolving signal patterns. 
Building on this, the Enhanced CHO (ECHO) scheme in \cite{prado2021echo} incorporates user trajectory prediction to proactively prepare 
base stations (BSs) along the anticipated user path, thereby improving handover target selection. 
Unlike these methods that directly optimize handover decisions, the work in \cite{sun2025proactive} 
develops an reinforcement learning-based approach to optimize key handover control parameters, such as 
time-to-trigger and cell individual offsets.
While these approaches are effective in TNs, they are unable to address the efficiency issues 
arising from the fundamental architectural and operational disparities between TNs and LSNs.

\section{Preliminary and Motivation} \label{section:motivation}
\begin{figure}[t]
\centering
\subfloat[Terrestrial Networks]{\includegraphics[width=0.45\textwidth]{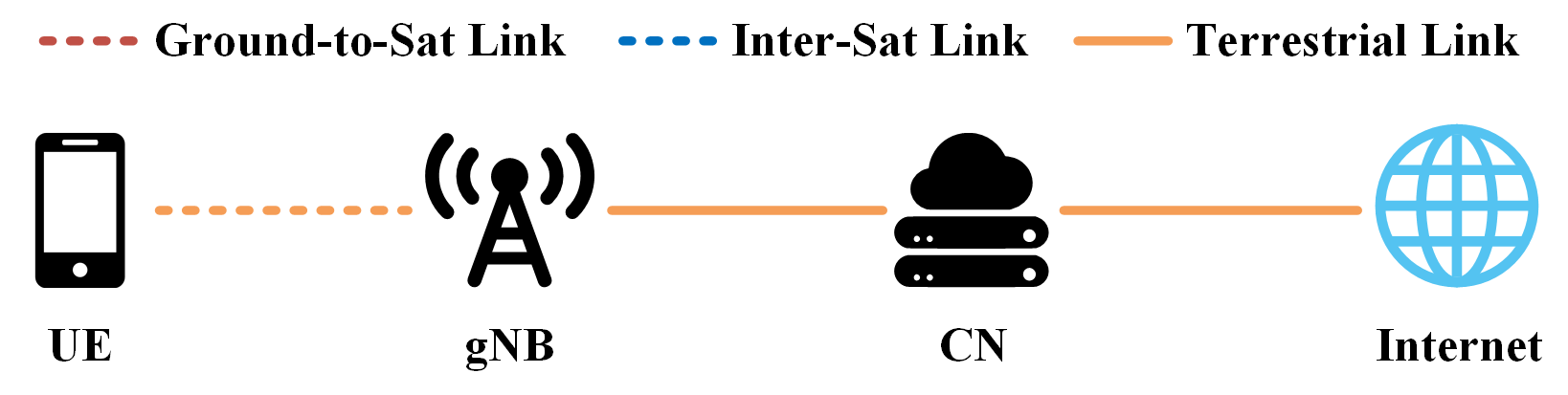} \label{fig:arch1}}
\vfil
\subfloat[Satellite Networks: Transparent Mode]{\includegraphics[width=0.45\textwidth]{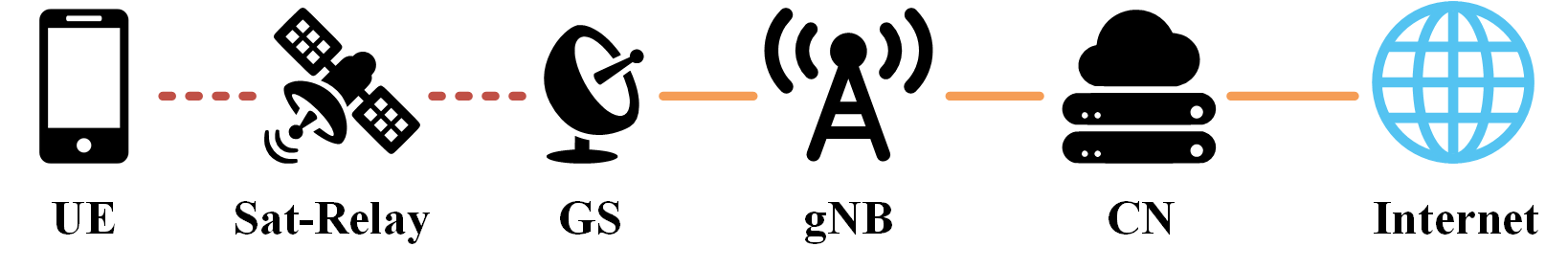} \label{fig:arch2}}
\vfil
\subfloat[Satellite Networks: Regenerative Mode]{\includegraphics[width=0.45\textwidth]{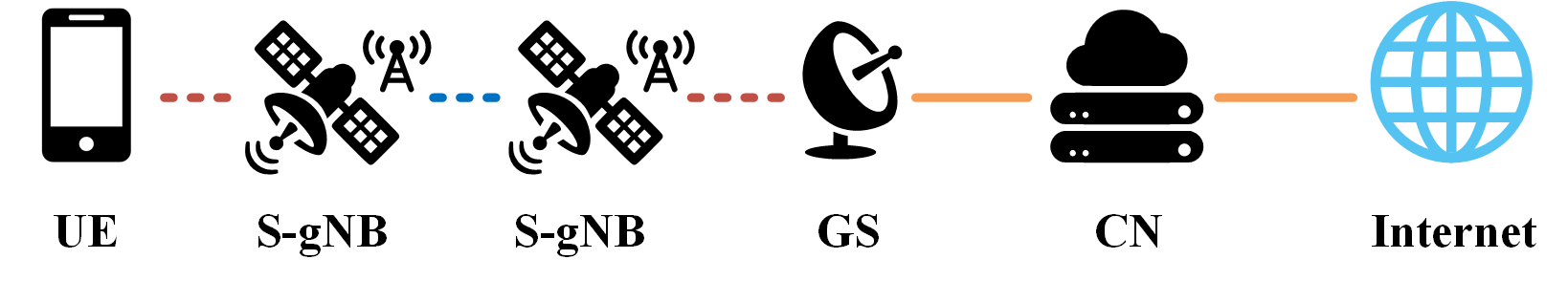} \label{fig:arch3}}
\caption{Architecture of Terrestrial and Satellite Networks.}
\label{fig:arch}
\end{figure}

In this section, we will provide a concise overview of the system architecture of current terrestrial and satellite networks.
Then we will briefly introduce existing handover mechanisms
and discuss the issues that arise when applying these handover mechanisms to LSNs.

\subsection{Architecture of Terrestrial and Satellite Networks}
As illustrated in Fig. \ref{fig:arch1}, 
Current 5G systems comprise three primary components: 
UEs, such as mobile phones and Internet of Things (IoT) devices, the Radio Access Network (RAN), and the CN. 
The principal entity in the RAN is the next-generation Node B (gNB), which serves as the 5G base station that connects UEs to the CN. 
The CN encompasses various network functions, with the two most relevant to handovers being the Access and Mobility Management Function (AMF) and the User Plane Function (UPF). 
The AMF manages user access, mobility, and security, while the UPF is responsible for forwarding user data packets and managing data paths.

In mobile satellite networks, UEs are connected to the CN via satellites. 
Currently, mobile satellite networks operate in two distinct modes \cite{3gpp.38.821}.
In the transparent mode, the satellite acts as a passive relay, simply forwarding signals between UEs and terrestrial gNBs through
Ground Stations (GSs) without altering the content.
In the regenerative mode, however, the satellite serves as a gNB and performs digital processing tasks such as error correction and signal decoding/encoding. 
This mode enhances signal quality and reduces latency, making it the preferred mode for future mobile satellite networks.
Therefore, this paper focuses on the regenerative mode and refers to these satellites as Satellite gNBs (S-gNBs).

\begin{figure}[t]
\centering
\includegraphics[width=0.45\textwidth]{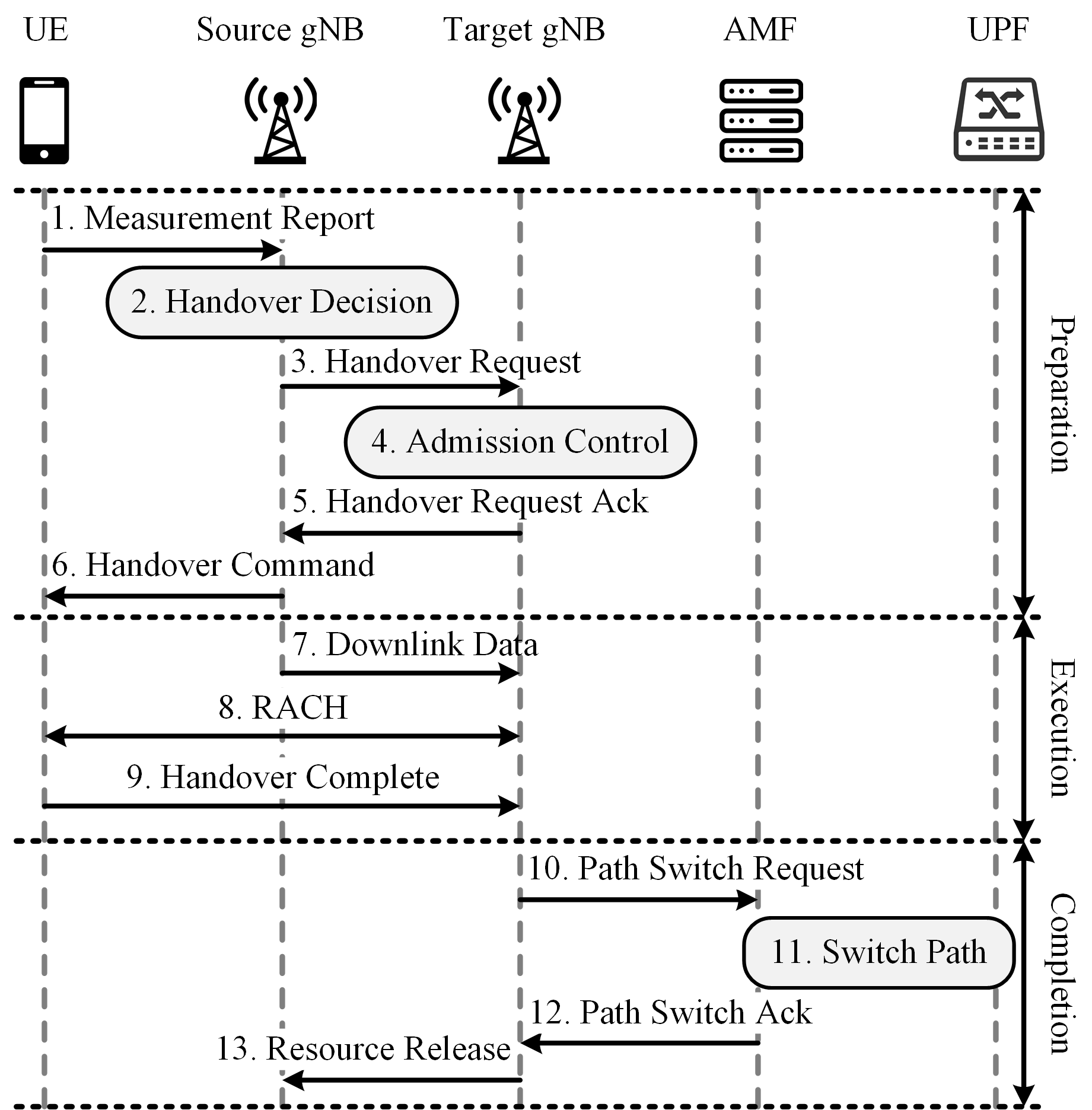}
\caption{Illustration of the BHO procedure.}
\label{fig:BHO}
\end{figure}

\subsection{Handover Mechanisms in TNs}
Fig. \ref{fig:BHO} illustrates the simplified baseline handover (BHO) procedure in 5G systems
where the source and target gNBs are assumed to be interconnected via the Xn interface \cite{3gpp.38.300}.
The process can be divided into the following three phases:
\begin{itemize}
\item \textbf{Preparation Phase:} When a UE detects that the signal from a new gNB is stronger than that of the source gNB by a specified margin, 
  the UE sends a Measurement Report (MR) to the source gNB. 
  The source gNB then selects a target gNB based on the MR and returns the related information about the target gNB to the UE.
\item \textbf{Execution Phase:} The UE disconnects from the source gNB and connects to the target gNB via the Random Access CHannel (RACH) process. 
  The UE sends a handover complete message when the connection is successfully established.
  Meanwhile, the source gNB forwards any undelivered downlink data to the target gNB.
\item \textbf{Completion Phase:} The target gNB notifies the AMF and UPF to switch the downlink path. 
  Subsequently, the target gNB instructs the source gNB to release the resources associated with the UE.
\end{itemize}
A notable issue in BHO is that the handover process is typically initiated when the radio link between the UE and the source gNB is already weakened. 
Consequently, the MR may fail to reach the network, or even if it does, 
the handover command may not successfully reach the UE, resulting in handover failure.

To mitigate this problem, 3GPP introduced conditional handover (CHO),
which decouples the preparation phase from the execution phase \cite{martikainen2018basics}.
In CHO, the UE still sends an MR to initiate the handover process, but earlier than in BHO to prevent the radio link from becoming too weak.
At this preliminary stage, the source gNB may find it challenging to identify the optimal target gNB and therefore usually selects multiple candidates. 
The related information of these candidates, along with the handover execution conditions, is returned to the UE. 
The UE maintains its connection with the source gNB and simultaneously evaluates the execution conditions for all candidates.
When one of these conditions is met, the UE autonomously switches to the corresponding target gNB without further communication with the source gNB.


\subsection{Existing Problems in LEO Satellite Networks} \label{subsection:problems}
Both BHO and CHO are designed specifically for TNs. 
When directly applied to LSNs, in addition to the signal strength issue discussed in the Introduction section, 
they also present the following significant limitations:
\begin{itemize}
  \item \textbf{Excessive Signaling Overhead:}
    The rapid movement and broad coverage of LEO satellites necessitate frequent handovers for a large number of UEs.
    Handling such a massive amount of handovers leads to unacceptable signaling overhead for the CN \cite{li2022case}.
  \item \textbf{Prolonged Handover Latency:}
    In LSNs, the communication latency between S-gNBs and the CN is significantly prolonged due to 
    the extended distance between them. 
    This situation contributes to an average handover latency of approximately $400$ ms \cite{wu2024accelerating}, 
    which detrimentally impacts user experience, particularly for latency-sensitive applications.
  \item \textbf{Suboptimal Handover Decisions:}
    In TNs, the selection of the target gNB is primarily based on signal strength. 
    However, in LSNs, additional factors, such as the remaining service time and workload of S-gNBs, significantly influence user experience. 
    Hence, it is necessary to optimize the handover decisions to account for these factors.
  \item \textbf{Unavailable Xn Interface:}
    Both BHO and CHO assume that gNBs are interconnected via the Xn interface. 
    While connecting S-gNBs using Inter-Satellite Links (ISLs) is feasible, this technology is still not fully mature. 
    Even if ISLs become viable, connecting satellites at different orbits would necessitate multiple hops and is inefficient.
    An alternative is to use N2-based handover, where gNBs communicate via the CN,
    but this approach considerably increases signaling overhead and handover latency.
\end{itemize}

As discussed in Section \ref{section:introduction}, the first two problems stem from the unpredictability of handovers in TNs
and are further exacerbated in LSNs due to the rapid movement, broad coverage,
and high communication latency of LEO satellites.
However, these issues can be largely mitigated in LSNs by leveraging the predictability of handovers, 
which significantly improves handover efficiency. 
Moreover, with adjusted system architecture and handover procedure, our proposed handover mechanism addresses the last two problems as well. 
The detailed mechanism design is presented in the next section.

\section{Mechanism Design} \label{section:system}

\begin{figure}[t]
\centering
\includegraphics[width=0.36\textwidth]{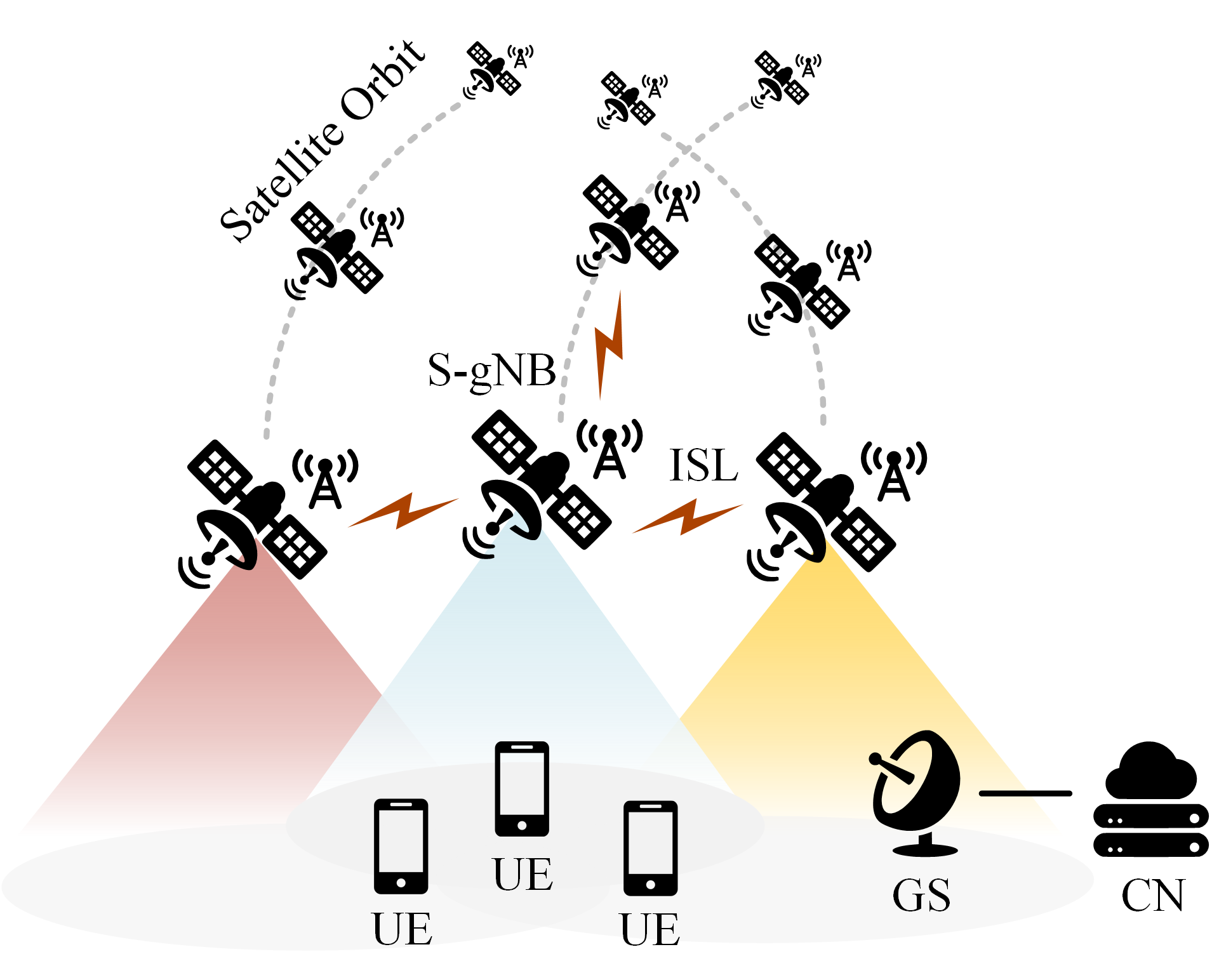}
\caption{Illustration of the considered scenario.}
\label{fig:system}
\end{figure}

In this section, we first present the proposed handover mechanism, termed PreHO, under idealized conditions where
predictive information---such as channel states and satellite trajectories---is assumed to be accurate.
Subsequently, we demonstrate that PreHO is fully compatible with existing handover mechanisms and can seamlessly revert to them 
when predictive information is inaccurate.

Before delving into the details of our handover mechanism, we briefly describe the scenario considered in this paper.
As illustrated in Fig. \ref{fig:system}, we consider a specific region on Earth where multiple UEs are served by a LSN.
UEs are considered relatively stationary compared to the high speed of LEO satellites.
We assume UEs are equipped with Global Navigation Satellite System (GNSS) capabilities so that they are aware of their own locations.

\subsection{Design Principles}
The PreHO mechanism leverages the predictability of handovers in LSNs, which offers two principal benefits.
First, it allows us to skip steps associated with the random movement of UEs, thereby substantially reducing signaling overhead and handover latency. 
Second, it enables us to optimize handover decisions in a coordinated manner, enhancing both user experience and system performance. 
Specifically, PreHO primarily differs from conventional handover procedures in the following four aspects:

\subsubsection{MR-less}
Traditional handover schemes rely on MRs to capture signal strength fluctuations induced by the 
random movement of UEs and variations in channel conditions.
However, in LSNs, the randomness of UEs' movement and channel states is largely mitigated,
hence the signal strength can typically be predicted with high accuracy given known user locations \cite{9439942, 10550141, 9990587}.
Motivated by this, PreHO eliminates the need for MRs by using predicted signal strength to determine handover timing.
Consequently, rather than sending MRs, UEs only need to periodically report their locations.

\subsubsection{Coordinated Handover Decisions}
Determining appropriate handover times and target S-gNBs in LSNs is inherently more challenging than in TNs due to the following three reasons.
First, with the ongoing development of LEO constellations, a UE is typically covered by tens or even 
hundreds of LEO satellites \cite{guo2024fast, del2019technical}, 
resulting in a substantially larger number of candidate target S-gNBs. 
Second, there are additional factors, such as the remaining service time of S-gNBs, 
need to be considered when making handover decisions. 
Third, for better system performance, handover decisions of different UEs should be globally coordinated to balance the load across S-gNBs. 
These factors make it very difficult to derive optimal handover decisions without a global view of the system.

To address these issues, we introduce a novel network function, the Handover Planning Function (HPF), 
which can be placed at the GS or other entities capable of interfacing with both S-gNBs and the CN.
The HPF collects UE locations via S-gNBs and combine this information with ephemeris data to optimize handover decisions in a coordinated way.
These decisions are then distributed to the S-gNBs, UEs, and the CN for subsequent execution. 
The algorithm for making handover decisions will be presented in Section \ref{section:algorithm}.

\subsubsection{Time-based CHO}
The handover decisions derived by the HPF specify both the target S-gNB and the exact time at which UEs should handover.
Inspired by CHO, PreHO leverages this predictive capability to reduce the signaling overhead and the probability of handover failure.
Specifically, after receiving handover decisions from the HPF, S-gNBs send the relevant information for multiple future handovers to corresponding 
UEs through a single signaling message. 
UEs then autonomously connect to the target S-gNB at the specified times, without further communication with the source S-gNBs. 
We refer to this new type of CHO as time-based CHO.

\subsubsection{RACH-less}
Conventional handover procedures require UEs to initiate a RACH procedure when connecting to the target gNB, 
primarily to obtain Timing Advance (TA) and uplink transmission grants.
The TA value corresponds to the time a signal takes to travel from a UE to the gNB and is used in uplink synchronization to avoid 
collisions and interference from different UEs. 
In LSNs, since the locations of UEs and satellites are known, the TA value can be effectively estimated \cite{wang2021location, balakrishnan2025timing}. 
Moreover, the uplink transmission grant can be transmitted to the UE in advance along with the handover decisions. 
As a result, we can remove the RACH process to reduce the handover latency.

\subsection{Procedure of PreHO}
\begin{figure}[t]
\centering
\includegraphics[width=0.48\textwidth]{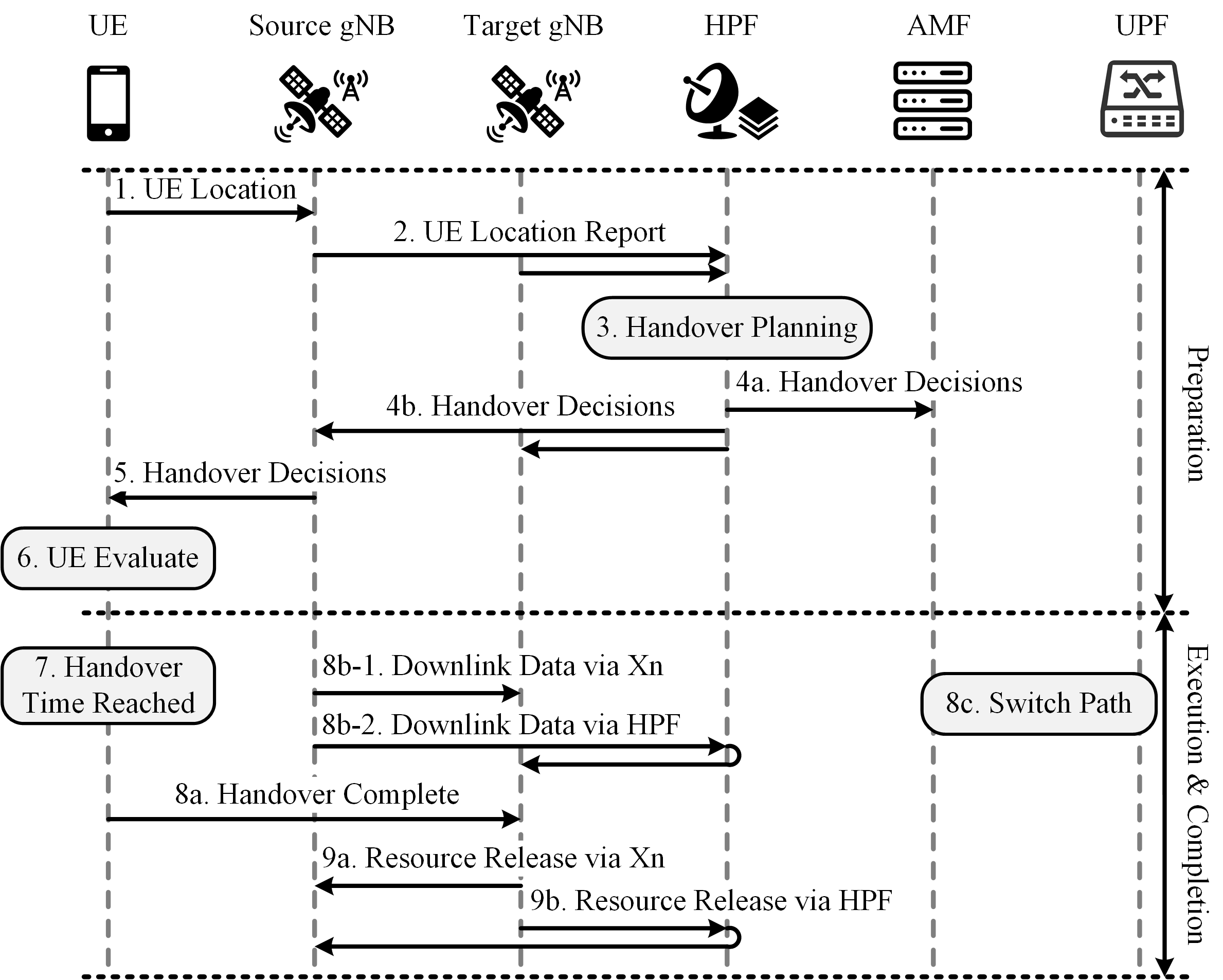}
\caption{Illustration of the PreHO procedure.}
\label{fig:PreHO}
\end{figure}

Based on the above discussion, we can outline the procedure of PreHO.
The primary steps are illustrated in Fig. \ref{fig:PreHO}.
For illustrative clarity, we first assume the predictive information is accurate.
Then in the next subsection, we show that PreHO is compatible with existing handover mechanisms
and can seamlessly revert to them when the predictive information is inaccurate.

The time horizon is divided into planning intervals of equal length. 
Before each planning interval, UEs report their locations to the serving S-gNB, which then forwards this information to the HPF.
With UE locations and ephemeris data, the HPF derives the optimal handover decisions of UEs for the upcoming planning interval.
The HPF also estimates the TA value and allocates uplink transmission grants for every planned handover.
The obtained handover decisions, along with corresponding TA values and uplink transmission grants, are distributed to the relevant
UEs, S-gNBs, and the CN for subsequent execution.

During the planning interval, each UE will automatically connect to the appropriate S-gNB at the specified time according to the handover decisions. 
Meanwhile, the CN autonomously switches the downlink path to the correct S-gNB.
If there are undelivered downlink packets buffered at the source S-gNB,
these packets will be forwarded to the target S-gNB via the Xn interface or the HPF, depending on whether there is an ISL path between the source and target S-gNB.
The resource release message is transmitted from the target S-gNB to the source S-gNB in a similar way.

\subsection{Compability with Existing Mechanisms}
In practical deployments, predictive information may deviate from actual network conditions due to various factors.
For example, UEs may leave or join the network during a planning interval,
channel states may change suddenly due to obstruction or weather,
and satellites may deviate from pre-planned trajectories to avoid collisions with space debris \cite{zhao2023first}.
However, since both UEs and S-gNBs are aware of actual network conditions, 
PreHO can seamlessly revert to traditional mechanisms whenever inconsistencies are detected.

To enable this fallback, the HPF transmits expected signal strength profiles alongside handover decisions (during Steps 4b and 5 in Fig. \ref{fig:PreHO}). 
UEs subsequently monitor real-time signal strength. 
If the actual signal from the source S-gNB degrades prematurely or the target S-gNB signal is below expected thresholds at the planned handover time, 
the UE can autonomously revert to signal strength-based handover by issuing an MR. 
The source S-gNB is then responsible for informing both the target S-gNB and CN to disregard the pre-planned handover for that UE. 
A similar fallback is triggered if UEs or satellites significantly deviate from their predicted locations.
In this sense, PreHO can be regarded as an orthogonal solution to the existing handover machanisms.


\subsection{Advantages of PreHO}
\textbf{Reduced signaling overhead.} 
In PreHO, the relevant information for multiple handovers within each planning interval is disseminated to UEs, gNBs, and the CN 
through a single signaling interaction. 
Consequently, the number of signaling messages is substantially reduced, significantly alleviating the signaling overhead. 
This reduction in signaling messages, coupled with a simplified handover procedure (e.g., RACH-less),
also reduces the power consumption associated with handovers, which is non-trivial in existing systems \cite{hassan2022vivisecting}.

\textbf{Improved handover latency.} 
In BHO, the duration of the execution phase is defined as the handover interruption time (HIT),
as UEs cannot transmit or receive any data during this period.
Similarly, the duration of the path switch process is defined as the downlink switching time (DST).
The handover latency refers to the sum of HIT and DST.
According to the data in \cite{iqbal2024rach, 3gpp.38.331, 3gpp.38.133},
eliminating the RACH process will reduce the HIT from 59 ms to 38 ms, resulting in a 35.6\% improvement. 
Given the extended communication latency between LEO satellites and UEs, this improvement is expected to be even more pronounced in LSNs.
Moreover, as outlined in the previous subsection, the handover decisions are transmitted to the CN in advance, 
enabling the CN to autonomously switch the downlink path without communicating with S-gNBs. 
Consequently, the completion phase is integrated into the execution phase, thereby eliminating the DST
and reducing the handover latency.

\textbf{Optimized handover decisions.} A key distinction between PreHO and existing handover mechanisms is the introduction of the HPF, 
which makes optimal handover decisions in a coordinated manner based on global information. 
In this study, handover decisions are based solely on UE locations and ephemeris data. 
However, the HPF can be easily extended to incorporate
additional modules (e.g., weather-based channel modeling) to further enhance the performance of pre-planned handover decisions.




\section{Handover Planning Algorithm}  \label{section:algorithm}
This section presents a detailed description of the planning algorithm that optimizes handover decisions based on UE locations and ephemeris data.
Notice that we only make handover decisions for one planning interval at a time.


\subsection{System Model and Problem Formulation}
Suppose there are $N$ UEs in the considered region that are served by a LSN.
Based on the ephemeris data, we can identify all of the $M$ S-gNBs that may provide service to the UEs during the considered planning interval.
We further divide the planning interval into $T$ time slots of equal length.
The set of considered UEs, S-gNBs, and time slots are denoted by $\mathcal{N}$, $\mathcal{M}$, and $\mathcal{T}$, respectively.

Let $x_{ij}[t] \in \{0, 1\}$ be a binary user association variable indicating whether UE
$i\in\mathcal{N}$ is served by S-gNB $j\in\mathcal{M}$ at slot $t$.
A handover occurs when a UE switches from one S-gNB to another across consecutive time slots.
Therefore, we can calculate the number of handovers during a planning interval based on the values of $x_{ij}[t]$.
Specifically, if UE $i$ retains the same serving S-gNB at slots $t$ and $t+1$, we have
$\sum_{j\in\mathcal{M}} \left(x_{ij}[t+1] - x_{ij}[t] \right)^2 = 0$.
Conversely, if UE $i$ switches S-gNB at slot $t+1$,
we have $\sum_{j\in\mathcal{M}} \left(x_{ij}[t+1] - x_{ij}[t] \right)^2 = 2$.
Consequently, the total number of handovers during $T$ time slots can be expressed as
\begin{equation*}
  N_{\mathrm{HO}} = \sum_{t\in\mathcal{T}} \sum_{i\in\mathcal{N}} \sum_{j\in\mathcal{M}} \frac{1}{2} \left( x_{ij}[t+1] - x_{ij}[t] \right)^2.
\end{equation*}

The position of UE $i$ is represented as a three-dimensional Cartesian coordinate, denoted by 
$L^u_i = (L^{u,x}_i, L^{u,y}_i, L^{u,z}_i)$.
Notice that we have assumed the UE is relatively stationary during a planning interval
so its position does not change across time slots.
Similarly, based on the ephemeris data, the position of S-gNB $j$ at the beginning of slot $t$ can be easily derived and is denoted by
$L^s_j[t] = (L^{s,x}_j[t], L^{s,y}_j[t], L^{s,z}_j[t])$.
Using the method in \cite{pratt2019satellite}, one can easily calculate the elevation angle between UE $i$ and S-gNB $j$,
which is denoted by $\theta(L^u_i, L^s_j[t])$.
In current practice, it is required that the elevation angle between a UE $i$ and its serving S-gNB $j$
must be larger than a threshold $\theta_{min}$, which can be expressed as
\begin{equation}
  \theta \left( L^u_i, L^s_j[t] \right) \geq x_{ij}[t]\theta_{min}.
  \label{cons:elevation}
\end{equation}
Since the positions of UEs and S-gNBs are known,
this constraint actually specifies the set of S-gNBs that can provide service to UEs at each slot.
Define $\mathcal{M}_i[t] = \{ j\in\mathcal{M} \mid \theta \left( L^u_i, L^s_j[t] \right) \geq \theta_{min} \}$ as the set of S-gNBs
whose elevation angle from UE $i$ satisfies the required threshold, then constraint \eqref{cons:elevation} is equivalent to
\begin{equation*}
  x_{ij}[t] \leq I_{j\in\mathcal{M}_i[t]}, \quad \forall i\in\mathcal{N}, \forall j\in\mathcal{M}, \forall t\in\mathcal{T},
\end{equation*}
where $I_{j\in\mathcal{M}_i[t]}$ is the indicator function that returns $1$ if $j\in\mathcal{M}_i[t]$ and $0$ otherwise.

According to the Shannon capacity formula, if UE $i$ is served by S-gNB $j$ at slot $t$, then the amount of data that can be transmitted
between them is given by
\begin{equation*}
  D_{ij}[t] = \Delta t_{ij} B_{ij} \log_2 \left( 1 + \mathrm{SINR}_{ij}[t] \right),
\end{equation*}
where $\Delta t_{ij}$ denotes the transmission time allocated to UE $i$ by S-gNB $j$ during slot $t$,
$B_{ij}$ represents the spectrum bandwidth allocated to UE $i$ by S-gNB $j$ at slot $t$,
and $\mathrm{SINR}_{ij}[t]$ is the corresponding signal-to-interference-plus-noise ratio (SINR) between UE $i$ and S-gNB $j$.
In practice, $\mathrm{SINR}_{ij}[t]$ can be estimated using analytical modeling approaches \cite{guidotti2020non} or data-driven methods \cite{cerar2021machine}.

As wireless resources must be shared among UEs served by the same S-gNB, the total time and spectrum resources allocated by an S-gNB in each time slot
are constrained by the slot duration $\Delta t$ and the maximum available spectrum bandwidth $B^{max}_{j}$, respectively.
Moreover, since the transmitted data volume $D_{ij}[t]$ is linear with respect to both time and spectrum bandwidth, 
we introduce, for notational simplicity, a single decision variable $y_{ij}[t] \in [0,1]$ to represent the fraction of resources allocated to UE $i$ by S-gNB $j$ at slot $t$,
defined as
\begin{equation*}
  y_{ij}[t] = \frac{\Delta t_{ij} B_{ij}}{\Delta t B^{max}_j}.
\end{equation*}
We further define the maximum achievable throughput between UE $i$ and S-gNB $j$ at time slot $t$ as
\begin{equation*}
  D^{max}_{ij}[t] = \Delta t B^{max}_j \log_2 \left( 1 + \mathrm{SINR}_{ij}[t] \right).
\end{equation*}
Accordingly, the transmitted data can be expressed as $D_{ij}[t] = y_{ij}[t] D^{max}_{ij}[t]$,
and the aggregate throughput of UE $i$ at slot $t$ is given by:
\begin{equation*}
  D_{i}[t] = \sum_{j\in\mathcal{M}}y_{ij}[t] D^{max}_{ij}[t].
\end{equation*}

It is important to emphasize that the variable $y_{ij}[t]$ does not represent an explicit scheduling decision. 
Rather, it is introduced to estimate the utility of each UE under given user association status and thereby facilitate the optimization of handover decisions. 
The actual wireless resource allocation (e.g. allocation of resource blocks in 4G/5G networks) is performed in real time after the handover decisions are made, 
based on more up-to-date information, such as instantaneous channel conditions and current user traffic demands.

Suppose the user experience under given throughput can be quantified by the utility function $u_i(D_{i}[t])$, then the overall system utility is given by
\begin{equation*}
  U_{\mathrm{UE}} = \sum_{t\in\mathcal{T}} \sum_{i\in\mathcal{N}} u_i\left( D_{i}[t] \right).
\end{equation*}
In practice, it is commonly assumed that the utility function $u_i(\cdot)$ satisfies the following properties \cite{kelly1998rate}:
\begin{assumption}
  The utility function $u_i(\cdot)$ is non-decreasing and concave.
  \label{assumption}
\end{assumption}
\noindent
A widely adopted family of utility functions satisfying this assumption is the $\alpha$-fairness class \cite{xinying2023guide}, 
given by:
\begin{equation}
  u^{\alpha}_i(D_i[t]) = 
    \begin{cases}
      \frac{D_i[t]^{1-\alpha}}{1-\alpha}, & \mbox{for } \alpha \geq 0, \alpha \neq 1 \\
      \log(D_i[t]), & \mbox{for } \alpha = 1.
    \end{cases}
    \label{eq:alpha_fairness}
\end{equation}
These utility functions form a continuum that stretches from the utilitarian criterion ($\alpha = 0$) to the maximin criterion ($\alpha \to \infty$),
enabling a tunable trade-off between fairness and efficiency.
In the following, we will use $u^{\alpha}_i(\cdot)$ as an illustrative example,
but our algorithm is applicable to any utility function satisfying Assumption \ref{assumption}.

Based on the above system model, the coordinated handover decision problem can be formulated as:
\begin{align}
  \min_{\bm{x},\bm{y}}\quad & N_{\mathrm{HO}} - \gamma U_{\mathrm{UE}} \label{problem} \\
  s.t.\quad & x_{ij}[t] \leq I_{j\in\mathcal{M}_i[t]}, \quad \forall i\in\mathcal{N}, \forall j\in\mathcal{M}, \forall t\in\mathcal{T} \quad \tag{\ref{problem}{a}} \label{problem:min_elevation} \\
            & \sum_{j\in\mathcal{M}} x_{ij}[t] = 1, \quad \forall i\in\mathcal{N}, \forall t\in\mathcal{T} \tag{\ref{problem}{b}} \label{problem:service_continuity} \\
            & \sum_{i\in\mathcal{N}} y_{ij}[t] \leq 1, \quad \forall j\in\mathcal{M}, \forall t\in\mathcal{T} \tag{\ref{problem}{c}} \label{problem:resource_allocation} \\
            & x_{ij}[t]\in\{0,1\}, 0 \leq y_{ij}[t] \leq x_{ij}[t], \notag \\
            & \qquad \qquad \qquad \forall i\in\mathcal{N}, \forall j\in\mathcal{M}, \forall t\in\mathcal{T} \tag{\ref{problem}{d}} \label{problem:variable}
\end{align}
where $\bm{x}$ and $\bm{y}$ are the collections of variables $x_{ij}[t]$ and $y_{ij}[t]$,
$\gamma$ is the weighting factor that balances the trade-off between the total number of handovers and the overall utility of UEs.
Constraint \eqref{problem:min_elevation} guarantees that if UE $i$ is served by S-gNB $j$ at slot $t$, then the elevation angle must be larger
than the required threshold $\theta_{min}$.
Constraint \eqref{problem:service_continuity} is the service continuity constraint that ensures every UE will be served by exactly one S-gNB
at each time slot.
Constraint \eqref{problem:resource_allocation} states that the allocated wireless resources cannot exceed the capacity of each S-gNB.
Constraint \eqref{problem:variable} specifies the feasible domain for decision variables, with
the condition $y_{ij}[t] \leq x_{ij}[t]$ ensuring that a UE only receives resources from its serving S-gNB.

The formulated problem \eqref{problem} is a mixed-integer programming problem, which is computationally intractable in general.
In fact, it can be shown that the problem remains NP-hard even in the simplified case with a single time slot, as stated in the following theorem:
\begin{theorem}
  Problem \eqref{problem} is NP-hard even if $T=1$.
  \label{theorem:1}
\end{theorem}
\begin{proof}
  %
  For a single time slot, there is no handover so the objective of problem \eqref{problem} is equivalent to $\max_{\bm{x}, \bm{y}} U_{\mathrm{UE}}$.
  Clearly, the optimal resource allocation of S-gNB $j$ only depends on the set of UEs it serves, 
  denoted by $\mathcal{N}_j = \{ i\in\mathcal{N} \mid x_{ij} = 1 \}$.
  Notice that we have dropped the time slot index $t$ for brevity.
  Consequently, the utility of UE $i\in\mathcal{N}_j$ can be expressed as a function of $\mathcal{N}_j$, denoted by $u_i(\mathcal{N}_j)$.
  Define $u_j(\mathcal{N}_j) = \sum_{i\in\mathcal{N}_j} u_i(\mathcal{N}_j)$ as the sum of utilities of all UEs served by S-gNB $j$,
  then $U_{\mathrm{UE}} = \sum_{j\in\mathcal{M}} u_j(\mathcal{N}_j)$.
  Since $u_i(D_{ij})$ is non-decreasing and concave, it is easy to verify that $u_j(\mathcal{N}_j)$ is a monotone submodular set function.
  As a result, problem \eqref{problem} is equivalent to a matroid constrained submodular maximization problem, which is known to be NP-hard
  \cite{qu2019distributed,calinescu2007maximizing}.
\end{proof}

\subsection{Algorithm Design} \label{subsection:algorithm_design}
Theorem \ref{theorem:1} indicates that directly solving problem \eqref{problem} is computationally intractable even for a significantly simplified scenario.
To obtain near-optimal solutions efficiently, this subsection develops a heuristic algorithm based on alternating optimization and dynamic programming.
Our approach is to first show that the optimal resource allocation $\bm{y}^*$ can be easily obtained for any given user association $\bm{x}$.
We then propose an efficient iterative algorithm that converges to the local optimum by optimizing the user association decisions for different
UEs in each step.

\subsubsection{Resource Allocation}
We first assume the user association decision $\bm{x}$ is given.
For convenience, let $\mathcal{N}_j[t]$ denote the set of UEs served by S-gNB $j$ at slot $t$, i.e. $\mathcal{N}_j[t] = \left\{ i\in\mathcal{N} \mid x_{ij}[t] = 1 \right\}$.
Then for any S-gNB $j\in\mathcal{M}$ and time slot $t\in\mathcal{T}$, the corresponding resource allocation subproblem can be expressed as follows
\begin{align}
  \max_{\bm{y}_j[t]} \quad & U_{\mathrm{UE}}^j[t] = \sum_{i\in\mathcal{N}_j[t]} u_i \left( D_{ij}[t] \right) \label{subproblem} \\
  s.t.\quad & \sum_{i\in\mathcal{N}_j[t]} y_{ij}[t] \leq 1 \tag{\ref{subproblem}{a}} \label{subproblem:resource_allocation} \\
            & 0 \leq y_{ij}[t] \leq 1, \quad \forall i\in\mathcal{N}_j[t] \tag{\ref{subproblem}{b}} \label{subproblem:variable}
\end{align}
where $\bm{y}_j[t]$ is the collection of variables $y_{ij}[t]$ for $i\in\mathcal{N}_j[t]$,
$D_{ij}[t] = y_{ij}[t] D^{max}_{ij}[t]$ is the amount of transmitted data.
Since $u_i(\cdot)$ is concave and the constraints are linear, problem \eqref{subproblem} is a convex optimization problem.
By applying KKT conditions, we can obtain the following optimality condition:
\begin{theorem}
  A feasible resource allocation $\bm{y}^*_j[t]$ is optimal if and only if there exists a constant $\lambda^*$ 
  such that $\frac{\partial u_i(D_{ij}[t])}{\partial y^*_{ij}[t]} = \lambda^*$
  for all $y^*_{ij}[t] > 0$ and $\frac{\partial u_i(D_{ij}[t])}{\partial y^*_{ij}[t]} \leq \lambda^*$ for all $y^*_{ij}[t] = 0$.
  \label{theorem:subproblem}
\end{theorem}
\begin{proof}
  Since $u_i(\cdot)$ is non-decreasing, we can safely assume constraint \eqref{subproblem:resource_allocation} is tight.
  In addition, we can also omit $y_{ij}[t] \leq 1$ in constraint \eqref{subproblem:variable} without changing the feasible region.
  Let $\lambda^*$ and $\mu_i^*$ be the multipliers corresponding to constraint \eqref{subproblem:resource_allocation} and \eqref{subproblem:variable}, respectively.
  Then according to the KKT conditions, the optimal $\bm{y}^*_j[t]$ must satisfy
  \begin{gather}
    \frac{\partial u_i(D_{ij}[t])}{\partial y^*_{ij}[t]} = \lambda^* + \mu_i^*, \quad \forall i\in\mathcal{N}_j[t] \label{eq:stationarity} \\
    \mu_i^* \geq 0, \quad \forall i\in\mathcal{N}_j[t] \label{eq:dual} \\
    \sum_{i\in\mathcal{N}_j[t]} \mu_i^* y^*_{ij}[t] = 0. \label{eq:slackness}
  \end{gather}
  According to \eqref{eq:dual} and \eqref{eq:slackness}, if $y^*_{ij}[t] > 0$, then we must have $\mu_i^* = 0$.
  Substituting into \eqref{eq:stationarity} yields $\frac{\partial u_i(D_{ij}[t])}{\partial y^*_{ij}[t]} = \lambda^*$.
  Similarly, if $y^*_{ij}[t] = 0$, then we have $\frac{\partial u_i(D_{ij}[t])}{\partial y^*_{ij}[t]} = \lambda^* + \mu^*_i \leq \lambda^*$.
\end{proof}
For a wide range of utility functions, we can directly obtain the analytical form
of the optimal $\bm{y}^*_j[t]$ based on Theorem \ref{theorem:subproblem}.
We take the $\alpha$-fairness utility functions given in \eqref{eq:alpha_fairness} as an example.
For convenience, we only consider the situation where $\alpha > 0$.
In this case, it can be easily verified that $y^*_{ij}[t] > 0$ for all $i\in\mathcal{N}_j[t]$, hence we have
\begin{equation}
  \frac{\partial u_i(D_{ij}[t])}{\partial y^*_{ij}[t]} = y^*_{ij}[t]^{-\alpha} D^{max}_{ij}[t]^{1-\alpha} = \lambda^*, \quad \forall i\in\mathcal{N}_j[t].
  \label{eq:y_optimal}
\end{equation}
Since $u_i(\cdot)$ is non-decreasing, we can safely assume constraint \eqref{subproblem:resource_allocation} is tight.
Substituting \eqref{eq:y_optimal} into $\sum_{i\in\mathcal{N}_j[t]} y^*_{ij}[t] = 1$ yields
\begin{equation}
  \lambda^* = \bigg( \sum_{i\in\mathcal{N}_j[t]} D^{max}_{ij}[t]^{\frac{1-\alpha}{\alpha}} \bigg)^{\alpha}.
  \label{eq:lambda_optimal}
\end{equation}
Combining \eqref{eq:lambda_optimal} with \eqref{eq:y_optimal}, we have
\begin{equation*}
  y^*_{ij}[t] = \frac{D^{max}_{ij}[t]^{\frac{1-\alpha}{\alpha}}}{\sum_{i\in\mathcal{N}_j[t]}D^{max}_{ij}[t]^{\frac{1-\alpha}{\alpha}}}.
\end{equation*}

\begin{algorithm}[t]
  \caption{Algorithm for Solving $\bm{y}^*_j[t]$}
    \label{alg:resource_allocation}
    \begin{algorithmic}[1]
      \STATE Initialization: $\lambda_l = \lambda_{min}, \lambda_r = \lambda_{max}$;
      \WHILE{$\lambda_r - \lambda_l \geq \epsilon$}
          \STATE $\lambda_m = \frac{1}{2}(\lambda_l + \lambda_r)$;
          \FOR{$i\in\mathcal{N}_j[t]$}
            \STATE Calculate $y_{ij}[t] = f(\lambda_m)$ using bisection search; \label{line:bisection}
          \ENDFOR
          \IF{$\sum_{i\in\mathcal{N}_j[t]} y_{ij}[t] > 1$}
              \STATE $\lambda_l = \lambda_m$;
          \ELSE
              \STATE $\lambda_r = \lambda_m$;
          \ENDIF
      \ENDWHILE
      \STATE $\lambda^* = \frac{1}{2}(\lambda_l + \lambda_r)$;
      \RETURN ${y}^*_{ij}[t] = f(\lambda^*)$
    \end{algorithmic}
\end{algorithm}

For more general utility functions that may not admit closed-form analytical expressions, 
the optimal solution $\boldsymbol{y}^*_j[t]$ can be efficiently obtained via a bisection search to determine the optimal dual variable $\lambda^*$. 
The overall procedure is summarized in Algorithm~\ref{alg:resource_allocation}.
At initialization, the left and right endpoints of the search interval, denoted by $\lambda_l$ and $\lambda_r$, 
are set to predefined lower and upper bounds on $\lambda^*$, namely $\lambda_{\min}$ and $\lambda_{\max}$, respectively. 
In each iteration, the midpoint $\lambda_m$ of the current interval is first computed. The loop in Lines~4-6 then seeks the value of $y_{ij}[t]$ that satisfies
\begin{equation}
  \frac{\partial u_i\bigl(D_{ij}[t]\bigr)}{\partial y_{ij}[t]} = \lambda_m.
  \label{eq:f}
\end{equation}
For notational convenience, we denote by $f(\lambda_m)$ the solution to \eqref{eq:f}. 
Since the utility function $u_i(\cdot)$ is concave, its derivative is non-increasing, which enables the use of a bisection method to efficiently compute $f(\lambda_m)$. 
Moreover, it follows that $f(\cdot)$ is also a non-increasing function.
Consequently, if the sum of the resulting $y_{ij}[t]$ values exceeds $1$, this indicates that $\lambda_m$ is smaller than the optimal value $\lambda^*$. 
In this case, the lower bound is updated as $\lambda_l = \lambda_m$, and the search proceeds over the upper half of the interval. 
Otherwise, the upper bound is updated as $\lambda_r = \lambda_m$. 
This iterative process continues until the length of the search interval falls below the prescribed accuracy threshold $\epsilon$.

It is worth noting that Algorithm~\ref{alg:resource_allocation} does not rely on the availability of a closed-form 
analytical expression for the utility function. 
Instead, it only requires the evaluation of $u_i(\cdot)$ and its derivative at given points. 
This property allows for a broad class of utility functions to be accommodated, including those parameterized by neural networks. 
Moreover, by employing a bisection search, the proposed algorithm achieves high computational efficiency.
Specifically, suppose that the accuracy requirement in Line~\ref{line:bisection} of Algorithm~\ref{alg:resource_allocation} is also $\epsilon$. 
Then, the time complexity of the inner loop spanning Lines~4-6 is 
$O\left(N_j[t]\log\left(\frac{1}{\epsilon}\right)\right)$, 
where $N_j[t] = |\mathcal{N}_j[t]|$ denotes the cardinality of the set $\mathcal{N}_j[t]$. 
Accounting for the outer bisection loop, the overall time complexity of the algorithm is therefore 
$O\left(N_j[t]\log^2\left(\frac{1}{\epsilon}\right)\right)$.

\subsubsection{User Association}
According to the above discussion, we can easily obtain the optimal resource allocation for any given user association $\bm{x}$, denoted by $\bm{y}^*(\bm{x})$.
Substituting into problem \eqref{problem} yields the following equivalent formulation:
\begin{align}
  \min_{\bm{x}}\quad & N_{\mathrm{HO}} - \gamma U_{\mathrm{UE}} \label{eq_problem} \\
  s.t.\quad & x_{ij}[t] \leq I_{j\in\mathcal{M}_i[t]}, \quad \forall i\in\mathcal{N}, \forall j\in\mathcal{M}, \forall t\in\mathcal{T} \tag{\ref{eq_problem}{a}} \label{eq_problem:min_elevation} \\
            & \sum_{j\in\mathcal{M}} x_{ij}[t] = 1, \quad \forall i\in\mathcal{N}, \forall t\in\mathcal{T} \tag{\ref{eq_problem}{b}} \label{eq_problem:service_continuity} \\
            & x_{ij}[t]\in\{0,1\}, \quad \forall i\in\mathcal{N}, \forall j\in\mathcal{M}, \forall t\in\mathcal{T} \tag{\ref{eq_problem}{d}} \label{eq_problem:variable}
\end{align}
where $U_{\mathrm{UE}}$ now only depends on $\bm{x}$.
The problem in \eqref{eq_problem} remains intractable because the utility of a given UE depends on the set of UEs associated with the same S-gNB. 
As a result, the user association decisions across different UEs are tightly coupled in $U_{\mathrm{UE}}$. 
To address this challenge, we adopt an alternating optimization strategy in which the user association decision of one UE is optimized at a time, 
while the association decisions of all other UEs are assumed to remain unchanged.

Our algorithm is described as follows.
We start from an arbitrary feasible user association decision.
At each iteration, we select a specific UE $i$ and fix the user association decisions of the rest UEs.
Let $U_i(t, j | \bm{x}[t])$ denote the maximum aggregate utility of all UEs at time slot $t$ when UE $i$ is 
served by S-gNB $j$ (i.e., $x_{ij}[t]=1$), while the associations of the remaining UEs are kept unchanged from $\bm{x}[t]$. 
For notational convenience, we denote the resulting user association by $\hat{\bm{x}}[t]$. 
Under this notation, $U_i(t, j | \bm{x}[t])$ can be expressed as
\begin{equation*}
  U_i(t, j | \bm{x}[t]) = \sum_{j'\in\mathcal{M}} \max_{\bm{y}_{j'}[t] \mid \hat{\bm{x}}[t]} U_{\mathrm{UE}}^{j'}[t],
\end{equation*}
where we have placed $\bm{y}_{j'}[t] | \hat{\bm{x}}[t]$ under the maximum operator to indicate 
that the optimization of $\bm{y}_{j'}[t]$ is performed given the user association $\hat{\bm{x}}[t]$. 
Clearly, the value of $U_i(t, j | \bm{x}[t])$ can be efficiently computed using Algorithm~\ref{alg:resource_allocation}. 
For brevity, we will omit the dependence on the initial association $\bm{x}[t]$ in the following text,
i.e., we use $U_i(t,j)$ instead of $U_i(t,j|\bm{x}[t])$.
Since the algorithmic description focuses on a single iteration, this will not cause any ambiguity.

Based on $U_i(t,j)$, we further define
\[
U_i([t_1,t_2], j) = \sum_{t=t_1}^{t_2} U_i(t,j),
\]
which represents the maximum aggregate utility of all UEs over the interval $[t_1,t_2]$ if UE $i$ is continuously served by S-gNB $j$. Accordingly, we define
\[
U_i([t_1,t_2]) = \max_{j \in \mathcal{M}} U_i([t_1,t_2], j),
\]
as the maximum aggregate utility over the same interval when UE $i$ is served by a single, but optimally chosen, S-gNB.

Let $C_i(t)$ denote the minimum objective value of problem~\eqref{eq_problem} over the first $t$ time slots, 
and let $\bm{x}^*_i$ be the corresponding optimal user association decision for UE $i$. 
If $\bm{x}^*_i$ incurs no handover during the first $t$ slots, i.e., UE $i$ is served by the same S-gNB throughout this period, then
\begin{equation}
  C_i(t) = N_{\mathrm{HO}} - \gamma U_{\mathrm{UE}} = -\gamma U_i([1,t]).
  \label{eq:C_case1}
\end{equation}
Otherwise, at least one handover occurs within the $t$ slots. Suppose that the last handover takes place at the end of slot $\tau$. 
In this case, the objective value can be expressed as
\begin{equation}
  C_i(t) = \min_{1 \leq \tau \leq t-1} \bigl\{ C_i(\tau) + 1 - \gamma U_i([\tau+1,t]) \bigr\}.
  \label{eq:C_case2}
\end{equation}
For notational convenience, we set $C_i(0)=0$ and introduce the indicator function $I_{\tau \geq 1}$, which equals $1$ if $\tau \geq 1$ and $0$ otherwise. 
Then, the update rules in \eqref{eq:C_case1} and \eqref{eq:C_case2} can be unified into the following expression:
\begin{align}
  C_i(t) = \min_{0 \leq \tau \leq t-1} \Bigl\{ C_i(\tau) + I_{\tau \geq 1} - \gamma U_i([\tau+1,t]) \Bigr\}.
  \label{eq:update_formula}
\end{align}

\begin{algorithm}[t]
  \caption{Algorithm for Solving $\bm{x}^*_i$}
    \label{alg:user_association}
    \begin{algorithmic}[1]
      \STATE Initialization: Set $C_i(0) = 0$ and calculate $U_i(t,j)$ for all $t\in\mathcal{T}, j\in\mathcal{M}$;
      \FOR{$t \in \{ 1,2, \dots, T\}$} \label{line:loop_start}
      \STATE Let $C_i(t) = - \gamma U_i([1,t])$;
        \FOR{$\tau \in \{ 1,\dots,t-1\}$}
        \STATE Calculate $C_i'(t) = C_i(\tau) + 1 - \gamma U_i([\tau+1, t])$;
        \IF{$C_i'(t) < C_i(t)$}
        \STATE $C_i(t) = C_i'(t)$;
        \ENDIF
        \ENDFOR
      \ENDFOR \label{line:loop_end}
      \STATE Obtain $\bm{x}^*_i$ by backtracking;
      \RETURN $\bm{x}^*_i$
    \end{algorithmic}
\end{algorithm}

Based on \eqref{eq:update_formula}, we can obtain $C_i(T)$ and the corresponding optimal user association $\bm{x}^*_i$ via dynamic programming.
The detailed steps are summarized in Algorithm \ref{alg:user_association}.
At the initialization step, we set $C_i(0) = 0$ and calculate the value of $U_i(t,j)$ for all $t\in\mathcal{T}$ and $j\in\mathcal{M}$.
Notice that obtaining the value of $U_i(t,j)$ requires solving the resource allocation subproblem \eqref{subproblem} for all $j\in\mathcal{M}$.
For situations where the analytical expression of $y^*_{ij}[t]$ is not available, the time complexity of obtaining $U_i(t,j)$
is $\sum_{j\in\mathcal{M}} O(N_j[t] \log^2(\frac{1}{\epsilon})) = O(N \log^2(\frac{1}{\epsilon}))$.
However, if the analytical expression of $y^*_{ij}[t]$ is available, the time complexity reduces to $O(N)$.
Hence, the total time complexity of the initialization step is $O(TMN \log^2(\frac{1}{\epsilon}))$ and $O(TMN)$, respectively.
The loop from line \ref{line:loop_start} to line \ref{line:loop_end} calculates $C_i(T)$ step by step based on \eqref{eq:update_formula}.
The time complexity for calculating $U_i([\tau+1, t])$ is $O(TM)$.
Hence, the time complexity of the loop is $O(T^3M)$.
The time complexity of the whole algorithm is the maximal value of the two processes.
After we have obtained $C_i(T)$, the corresponding optimal user association decision $\bm{x}^*_i$ can be derived by backtracking.

It should be emphasized that Algorithm~\ref{alg:user_association} optimizes the user association decision for a single UE only. 
To obtain a global solution, the algorithm must be iteratively applied to different UEs. 
Since the objective value is monotonically improved in each iteration, this iterative procedure is guaranteed to converge to a local optimum.
Simulation results presented in the next section demonstrate that the algorithm converges after each UE has been optimized once.

\section{Numerical Results} \label{section:simulation}
To evaluate the performance of our mechanism, we developed a prototype of PreHO using UERANSIM \cite{ueransim} and Open5GS \cite{5gs}. 
Specifically, UERANSIM is a widely adopted simulator for UEs and gNBs in 5G networks, while Open5GS is an open-source implementation of the 5G CN. 
The prototype was executed on a commodity personal computer equipped with eight 2.5 GHz CPU cores and 16 GB of RAM.

In the remainder of this section, we demonstrate the performance of PreHO from two aspects. 
First, we illustrate that PreHO can significantly reduce handover latency and signaling overhead. 
Second, we show that the proposed planning algorithm substantially improves user experience compared to benchmark schemes.

\begin{figure}[t]
\centering
\includegraphics[width=0.35\textwidth]{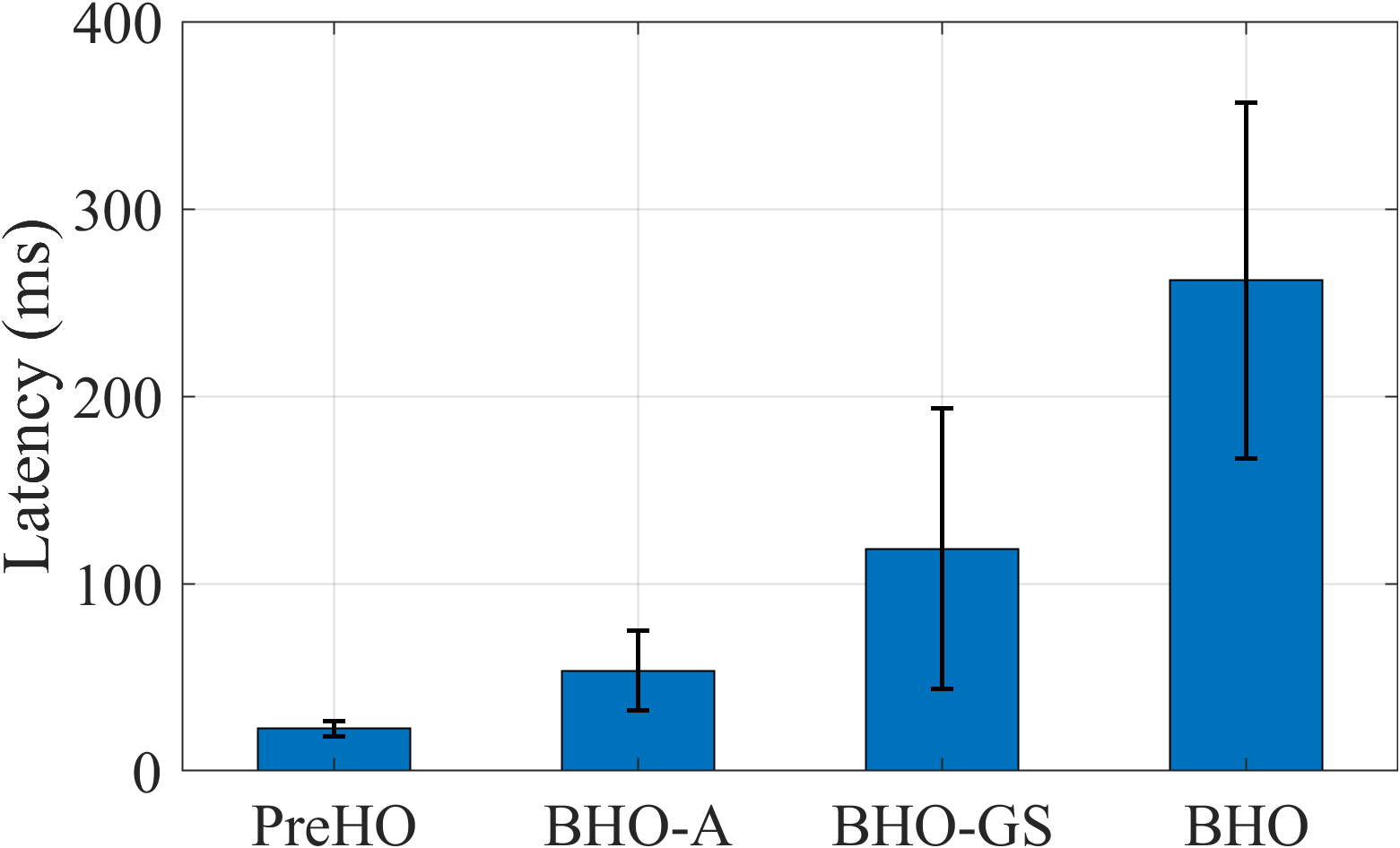}
\caption{Average handover latency in the Starlink constellation.}
\label{fig:average}
\end{figure}

\begin{figure*}[t]
\centering
\subfloat[CDF, Starlink, Similar Direction]{\includegraphics[width=0.32\textwidth]{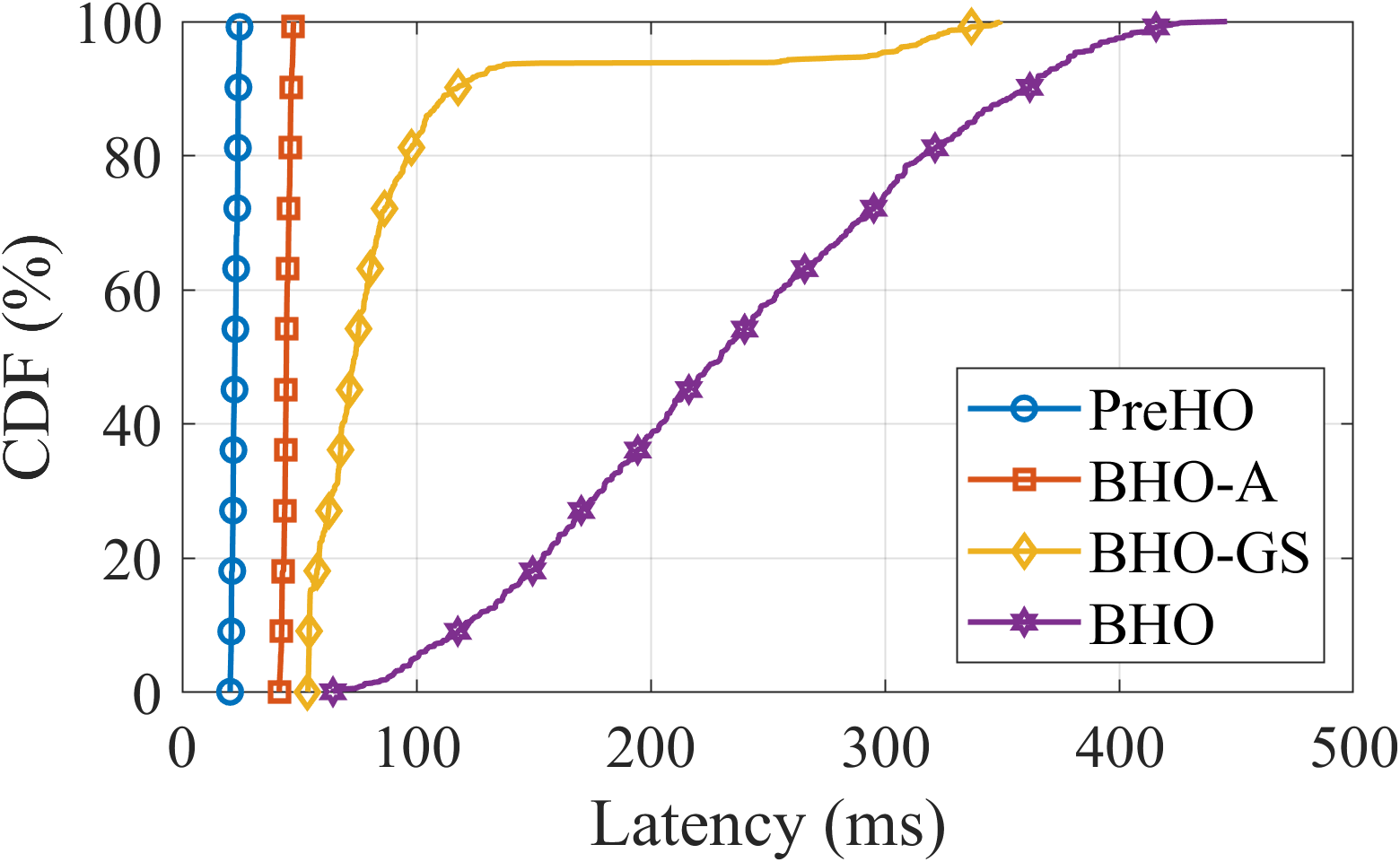} \label{fig:starlink-0}}
\hfil
\subfloat[CDF, Starlink, All Direction]{\includegraphics[width=0.32\textwidth]{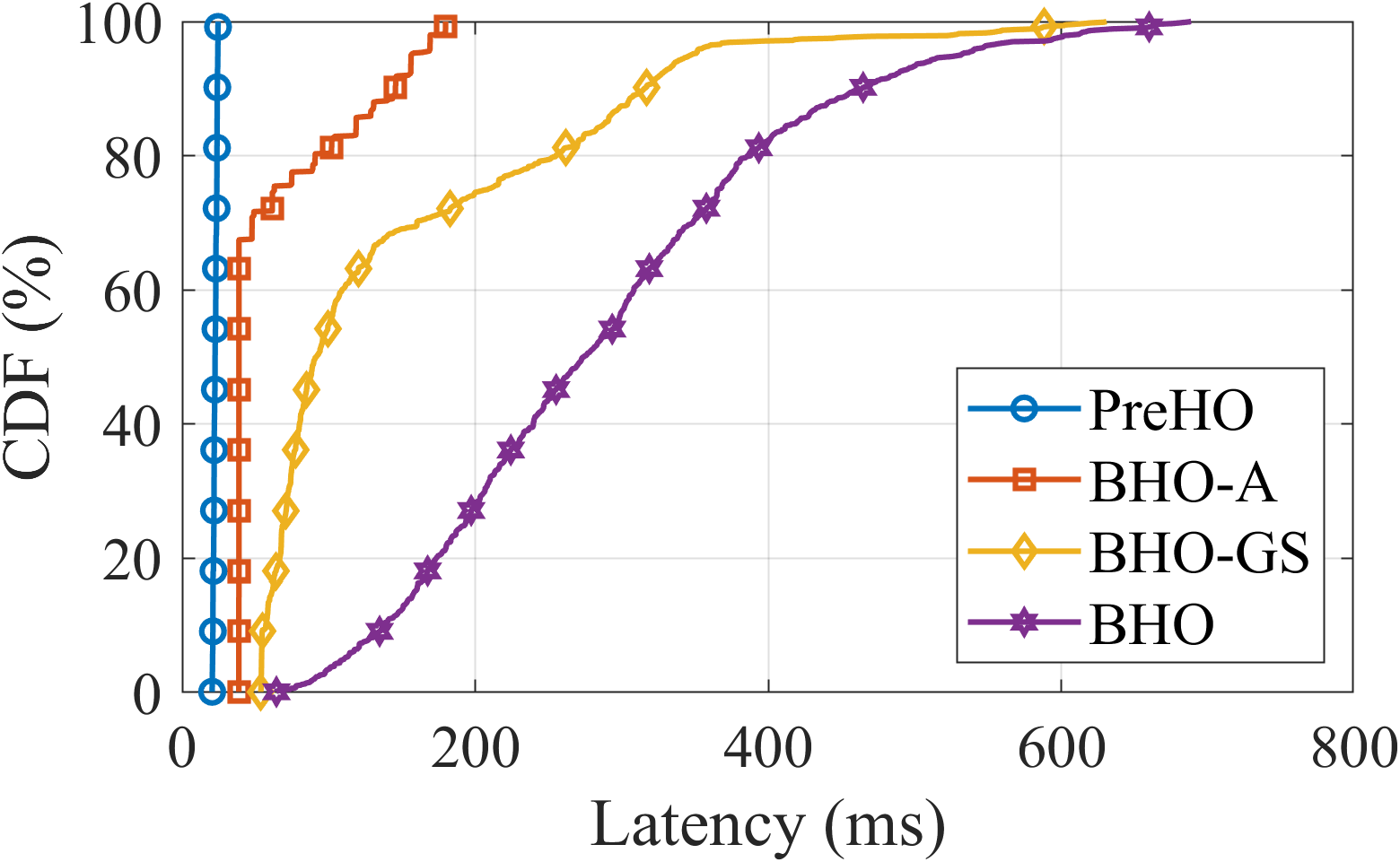} \label{fig:starlink-1}}
\hfil
\subfloat[CDF, Kuiper, Similar Direction]{\includegraphics[width=0.32\textwidth]{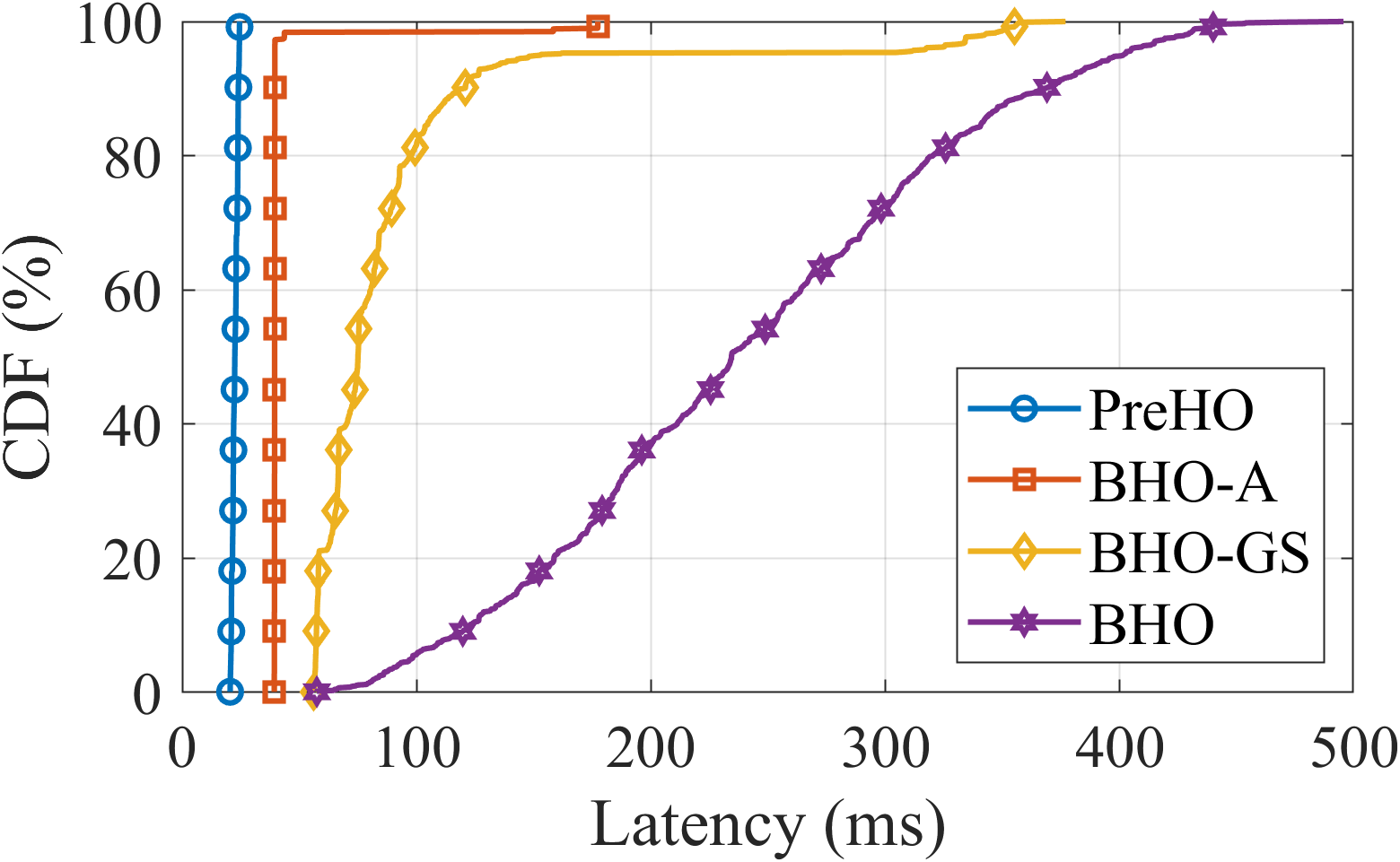} \label{fig:kuiper-0}}
\caption{Average and CDF of handover latency under different target S-gNB selection strategies and constellations.}
\label{fig:latency}
\end{figure*}



\subsection{Handover Procedure}
In our experiment, we selected one thousand UEs randomly distributed across the Earth, excluding polar regions. 
Satellite ephemeris and GS locations are obtained from the publicly available data of the Starlink and Kuiper constellations. 
In addition to the BHO, we also compare the performance of PreHO with two handover mechanisms designed for LSNs:
\begin{itemize}
  \item \textbf{BHO-GS} \cite{han2016distributed}: This mechanism conducts handovers with the assistance of nearby GSs, eliminating 
    the need to transmit signaling messages to the CN.
  \item \textbf{BHO-A} \cite{wu2024accelerating}: This accelerated handover mechanism allows the CN to predict handover decisions and autonomously switch the downlink path.
\end{itemize}

As described in \cite{wu2024accelerating}, there are two different strategies when selecting the target S-gNB.
The ``similar direction strategy'' requires the target S-gNB moves in the same direction as the source S-gNB (e.g., both moving from north to south).
In contrast, the ``all direction strategy'' imposes no such constraint on the direction of movement. 
The main difference is that the Xn interface for S-gNBs moving in opposite directions usually requires much more ISL hops,
leading to substantially higher communication latency \cite{chen2021analysis, zhang2022enabling}.

Fig. \ref{fig:average} presents the overall average handover latency of all schemes in the Starlink constellation.
Specifically, the average latency of PreHO is 22.5 ms, which is 11.6$\times$ shorter than the conventional BHO mechanism. 
With the assistance of nearby GSs, BHO-GS eliminates the communication latency between the GS and the CN, thereby reducing the average latency to 118.6~ms. 
However, this latency level remains insufficiently low to support latency-sensitive applications.
Similar to PreHO, BHO-A also avoids interactions between S-gNBs and the CN by predicting handovers at the CN.
However, BHO-A's procedure includes signaling interactions via the Xn interface, which
could result in considerable latency, particularly for S-gNBs moving in opposite directions. 
As a result, the average handover latency of BHO-A is 53.5 ms, which is approximately twice that of PreHO.

The Cumulative Distribution Function (CDF) of handover latency under various target S-gNB 
selection strategies and constellations is displayed in Fig. \ref{fig:latency}.
It is evident that PreHO consistently outperforms all benchmark methods across different scenarios. 
Comparing Fig. \ref{fig:starlink-0} and \ref{fig:starlink-1}, we observe a significant increase in handover latency 
for all benchmark algorithms, primarily due to the extended Xn interface latency. 
However, our mechanism remains unaffected as the downlink data can be transferred in advance.
Additionally, by comparing Fig. \ref{fig:starlink-0} and \ref{fig:kuiper-0}, 
we observe that the handover latency remains similar across different constellations.
This indicates that the variations in constellation configurations are not the primary factor affecting handover latency.

\subsection{Planning Algorithm}
To demonstrate the performance of our planning algorithm,
we simulate $N=100$ UEs randomly distributed within the region [35-38$^{\circ}$N], [122-125$^{\circ}$E].
The utilized constellation is the Starlink Group 1, which consists of 1584 LEO satellites at an altitude of 550 km,
with a required minimum elevation angle of 40$^\circ$.
The duration of each planning interval is set to 10 minutes, divided into $T=200$ slots, each lasting for 3 seconds.
Based on the above information, one can readily deduce that 
there are $M=39$ S-gNBs that may provide service to the designated region during the considered planning interval.
The SINR between UEs and S-gNBs follows the statistical model in \cite{guidotti2020non} 
and the maximum spectrum bandwidth of each S-gNB is $20$ MHz.
The utility function for each UE is defined as the $\alpha$-fairness function with $\alpha = 1$, where
the amount of transmitted data is measured in megabit (Mb).
In our experiments, the value of the weighting factor $\gamma$ is $2 \times 10^{-3}$.

Our planning algorithm is compared with the following three benchmarks:
\begin{itemize}
  \item \textbf{Largest Signal Strength (LSS)}: A UE initiates handover to a new S-gNB if its SINR is $50\%$ higher than that of the current S-gNB.
  \item \textbf{Longest Service Time (LST)}: A UE connects to the S-gNB with the longest remaining service time when the current S-gNB becomes unavailable.
  \item \textbf{Greedy}: In each slot, UEs greedily select the target S-gNB that optimizes the current objective value.
\end{itemize}

\begin{figure}[t]
  \centering
  \includegraphics[width=0.35\textwidth]{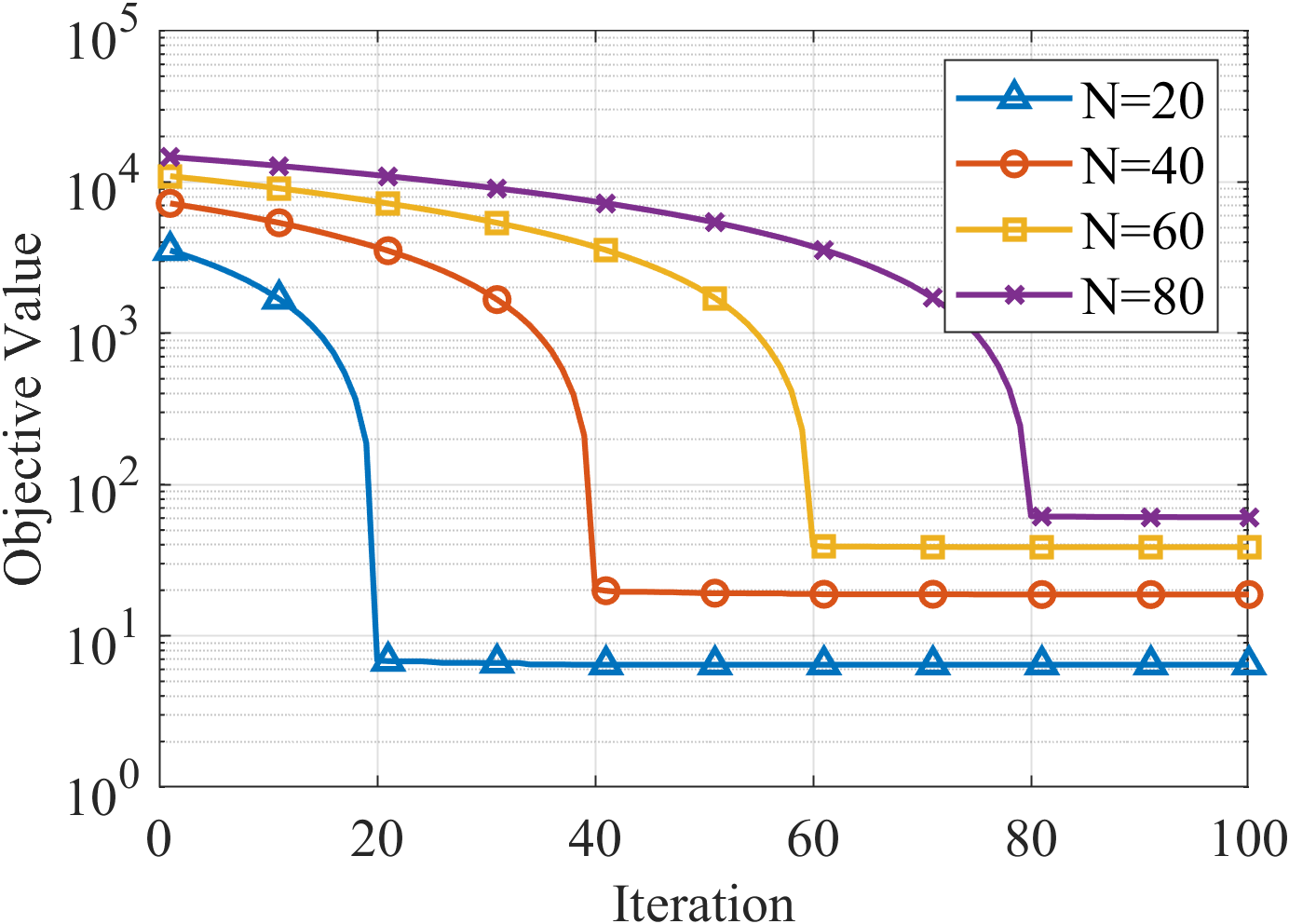}
  \caption{Convergence process of our planning algorithm.}
  \label{fig:convergence}
\end{figure}

\begin{figure}[t]
  \centering
  \includegraphics[width=0.35\textwidth]{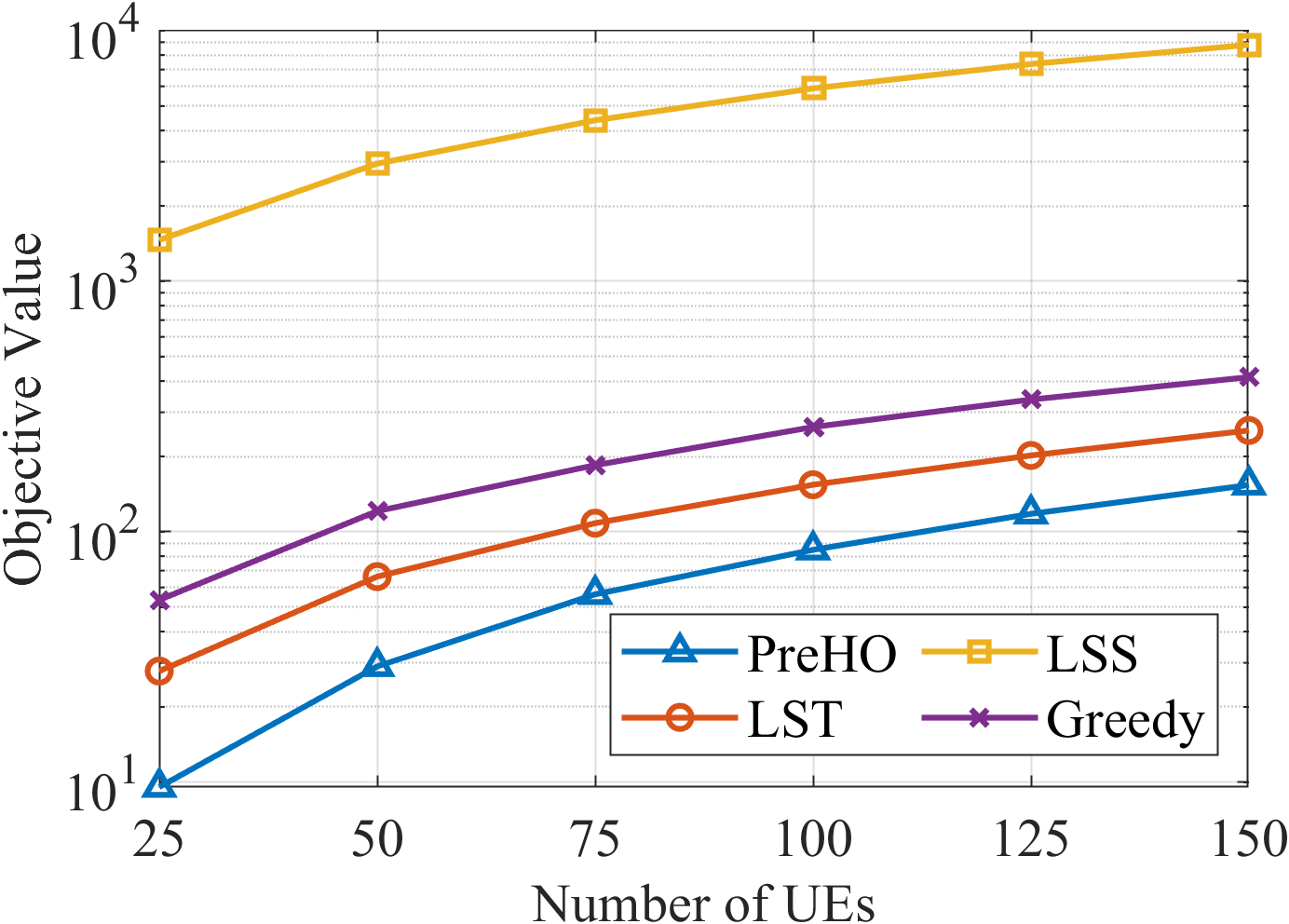}
  \caption{Performance under different numbers of UEs.}
  \label{fig:N}
\end{figure}

Recall that the proposed algorithm follows an iterative framework in which, at each iteration, the user association decision of a single UE is optimized.
Fig. \ref{fig:convergence} illustrates the convergence process of our algorithm under different numbers of UEs.
It can be observed that the improvement in the objective value becomes negligible once the number of iterations
reaches the number of UEs, i.e., each UE's handover decision has been optimized once.
Accordingly, in the subsequent experiments, the number of iterations is set equal to the number of UEs.

The comparison of all algorithms under different numbers of UEs is presented in Fig. \ref{fig:N}.
Notice that the y-axis is in log scale.
With more UEs, the total number of handovers increases and the average wireless resources allocated to each UE reduces,
hence the objective value increases.
The results indicate that PreHO consistently outperforms all benchmark methods in every scenario, 
while LST and Greedy demonstrate relatively lower performance. 
Notably, the LSS method, employed in conventional handover mechanisms, 
exhibits much worse performance compared to other approaches.
Specifically, PreHO achieves at least 57$\times$ improvement over LSS in all scenarios.

\begin{figure}[t]
  \centering
  \includegraphics[width=0.35\textwidth]{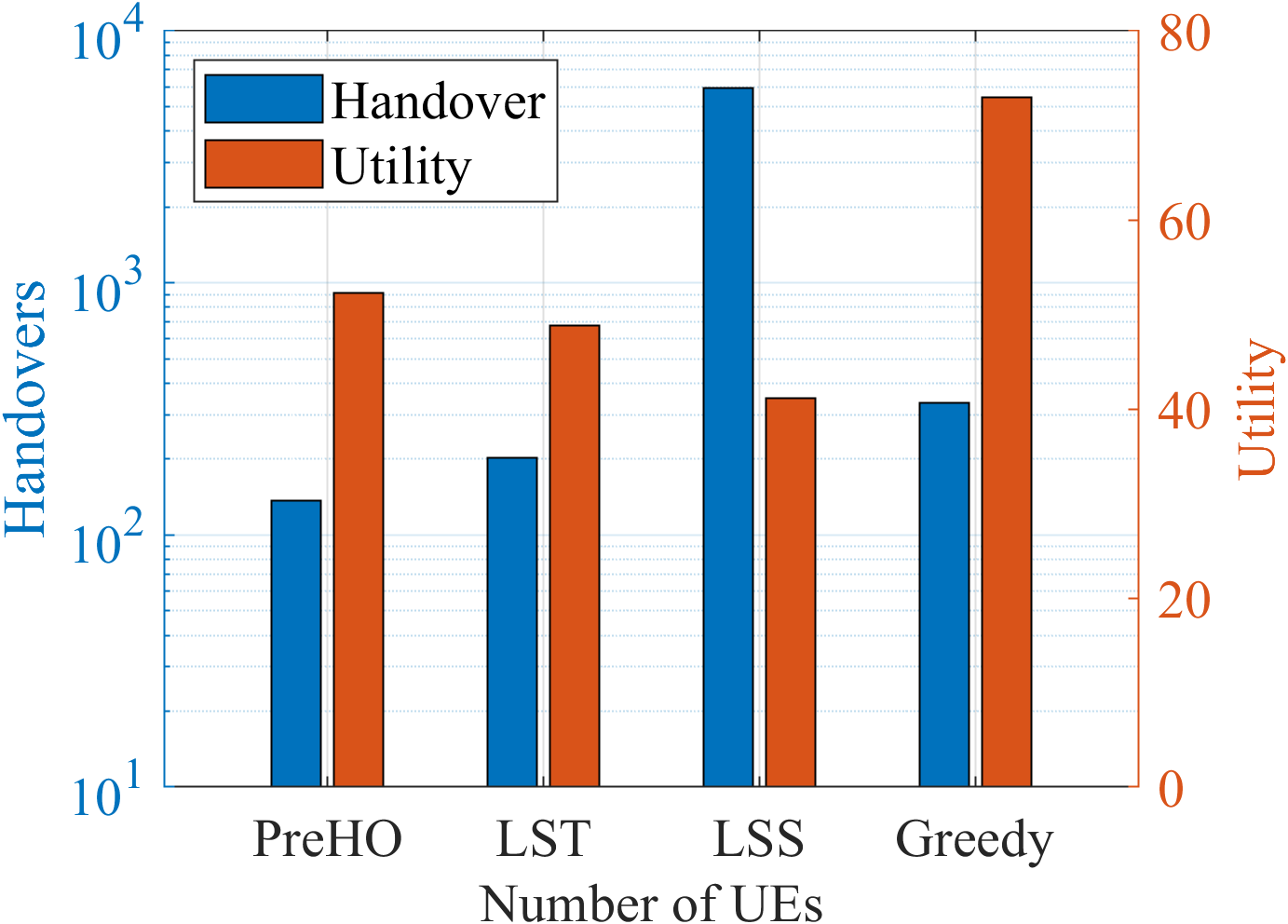}
  \caption{Number of handovers and UE utilities.}
  \label{fig:decomposition}
\end{figure}

\begin{figure}[t]
  \centering
  \includegraphics[width=0.35\textwidth]{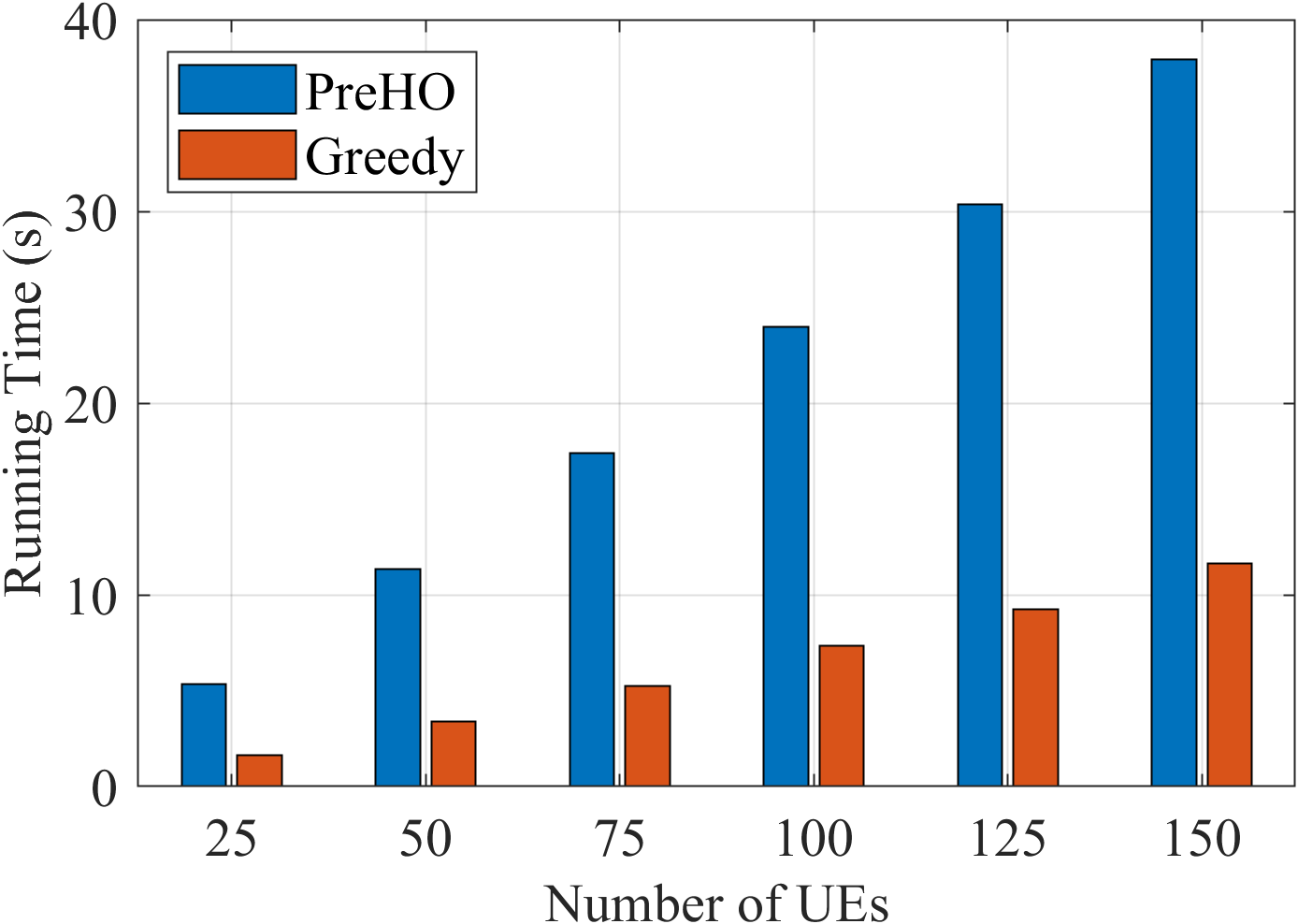}
  \caption{Running time under different numbers of UEs.}
  \label{fig:time}
\end{figure}

To further dissect the performance differences, Fig.~\ref{fig:decomposition} decomposes the objective into its two constituent components, 
namely, the total number of handovers and the aggregate utility weighted by $\gamma$. 
Note that the y-axis corresponding to the number of handovers is plotted on a logarithmic scale. 
Owing to the large number of S-gNBs covering the same geographical area and the randomness in channel gains, 
UEs adopting the LSS strategy tend to frequently switch their serving S-gNBs, giving rise to the well-known ping-pong effect. 
Consequently, the LSS scheme incurs a substantially higher number of handovers than the other three schemes, 
which constitutes the primary contributor to its elevated objective value. 
At the same time, LSS yields the lowest user utility, mainly due to excessive aggregation of UEs on high-SINR S-gNBs, resulting in severe load imbalance.

The Greedy scheme optimizes the objective value on a per-slot basis and therefore places greater emphasis on instantaneous user utility, 
achieving significantly higher utility than the other schemes. 
However, this myopic strategy also leads to a markedly larger number of handovers compared with PreHO and LST, 
which ultimately degrades the overall objective performance. 
By contrast, the LST scheme substantially reduces the number of handovers, albeit at the expense of reduced user utility. 
It is worth noting that the number of handovers under LST is slightly higher than that of PreHO. 
This is because LST progressively selects the S-gNB with the longest remaining service time starting from a randomly chosen initial association, 
whereas PreHO additionally optimizes the initial association decision. As a result, PreHO achieves a marginally lower number of handovers.

The simulation results also permit a comparative analysis of signaling overhead between PreHO and BHO at the CN. 
PreHO requires only a single signaling message per UE during each planning interval. 
In contrast, the LSS-based BHO approach results in 5,927 handovers within a single planning period, 
equating to approximately one handover every 10 seconds per UE. 
Although the evaluation of PreHO is conducted under an idealized setting in which all predictive information is assumed to be accurate,
the results nonetheless underscore its substantial potential in reducing signaling overhead.

Fig. \ref{fig:time} displays the running time of PreHO and Greedy.
The running time exhibits an approximately linear relationship with the number of UEs, remaining under 40 seconds for up to 150 UEs. 
This running time can be significantly reduced if we shorten the duration of planning intervals.
In addition, the PreHO is implemented in Python, which can be up to 100 times slower than lower-level languages such as C or C++ \cite{benchmark}.
Therefore, the running efficiency of PreHO can be considerably improved when deployed in practical systems.

\section{Conclusion} \label{section:conclusion}
To mitigate the frequent handovers and high signaling latency in LSNs,
this paper proposes PreHO, a predictive handover mechanism that transforms the traditional reactive handover paradigm into 
a proactive pre-planning mode by leveraging the predictable mobility patterns and stable channel states unique to LSNs. 
By integrating an MR-less and RACH-less design with time-based CHO procedures, 
the proposed mechanism significantly minimizes signaling overhead and handover interruption time.
Furthermore, the introduced HPF utilizes an efficient handover planning algorithm 
to achieve optimal balance between user utility and handover frequency.
Extensive evaluations demonstrate that PreHO substantially reduces handover latency and improves user experience,
offering a highly efficient solution for mobility management in next-generation large-scale LSNs.

\bibliographystyle{IEEEtran}
\bibliography{ref}

\end{document}